\documentclass[a4paper,10pt]{article}

\usepackage{amssymb}
\usepackage{amsmath}
\usepackage{subfigure}
\usepackage{verbatim}

\usepackage{color}
\usepackage{titlesec}

\titleformat{\paragraph}
{\normalfont\normalsize\bfseries}{\theparagraph}{1em}{}
\titlespacing*{\paragraph}
{0pt}{3.25ex plus 1ex minus .2ex}{1.5ex plus .2ex}

\allowdisplaybreaks
\usepackage{graphicx}
\usepackage{amsthm}
\usepackage{latexsym}
\usepackage[dvips]{epsfig}
\usepackage{wasysym}
\usepackage{mathrsfs}
\usepackage{eufrak}
\usepackage{bm}
\usepackage{authblk}
\usepackage{slashed}
\usepackage{yhmath} 
\usepackage{stmaryrd}
\usepackage{verbatim}

\theoremstyle{plain}
\newtheorem{proposition}{Proposition}
\newtheorem{lemma}{Lemma}
\newtheorem{theorem}{Theorem}
\newtheorem{assumption}{Assumption}
\newtheorem{conjecture}{Conjecture}

\newtheorem{remark}{Remark}

\def\bma{{\bm a}}
\def\bmb{{\bm b}}
\def\bmc{{\bm c}}
\def\bmd{{\bm d}}
\def\bme{{\bm e}}

\def\bmg{{\bm g}}
\def\bmh{{\bm h}}
\def\bmi{{\bm i}}
\def\bmj{{\bm j}}
\def\bmk{{\bm k}}

\def\bmx{{\bm x}}

\def\bmzero{{\bm 0}}
\def\bmone{{\bm 1}}

\def\bmthree{{\bm 3}}

\def\bmA{{\bm A}}
\def\bmB{{\bm B}}
\def\bmC{{\bm C}}
\def\bmD{{\bm D}}

\def\bmF{{\bm F}}

\def\bmH{{\bm H}}
\def\bmK{{\bm K}}

\def\bmP{{\bm P}}
\def\bmQ{{\bm Q}}

\def\bmX{{\bm X}}

\def\bmdelta{{\bm \delta}}

\def\bmphi{{\bm \phi}}

\def\bmGamma{{\bm \Gamma}}
\def\bmPhi{{\bm \Phi}}

\def\bmpartial{{\bm \partial}}

\def\bmhbar{{\bm \hbar}}

\newcounter{mnotecount}%[section]

\newcommand{\mnotex}[1]%{}
{\protect{\stepcounter{mnotecount}}$^{\mbox{\footnotesize $\bullet$\themnotecount}}$ 
\marginpar{%\color{red}%
\raggedright\tiny\em
$\!\!\!\!\!\!\,\bullet$\themnotecount: #1} }

\usepackage{titlesec}
\titleformat{\paragraph}
{\normalfont\normalsize\bfseries}{\theparagraph}{1em}{}
\titlespacing*{\paragraph}
{0pt}{3.25ex plus 1ex minus .2ex}{1.5ex plus .2ex}

\begin{document}

\title{\textbf{Polyhomogeneous expansions from time symmetric initial data}}

\author[,1]{E. Gasper\'in\footnote{E-mail address:{\tt
      e.gasperingarcia@qmul.ac.uk}}} \author[,1]{J. A. Valiente
  Kroon \footnote{E-mail address:{\tt j.a.valiente-kroon@qmul.ac.uk}}}
%\author[1]{Con T. Ributor}
\affil[1]{School of Mathematical Sciences, Queen Mary, University of
  London, Mile End Road, London E1 4NS, United Kingdom.}

\maketitle

\begin{abstract}
We make use of Friedrich's construction of the cylinder at spatial
infinity to relate the logarithmic terms appearing in asymptotic
expansions of components of the Weyl tensor to  the freely
specifiable parts of time symmetric initial data sets for the Einstein
field equations. Our analysis is based on the assumption that a
particular type of formal expansions near the cylinder at spatial
infinity corresponds to the leading terms of actual solutions to the
Einstein field equations. In particular, we show that if the Bach tensor of
the initial conformal metric does not vanish at the point at infinity
then the most singular component of the Weyl tensor decays near null
infinity as
$O(\tilde{r}^{-3}\ln \tilde{r})$ so that spacetime will not peel. We
also provide necessary conditions on the initial data which should lead to a
peeling spacetime. Finally, we show how to construct global spacetimes
which are candidates for non-peeling (polyhomogeneous) asymptotics.
\end{abstract}



%\setcounter{tocdepth}{2}
%\tableofcontents

\section{Introduction}
The \emph{Peeling theorem} has played a very important role in the
development of the modern notion of gravitational radiation. It is
usually formulated within the context of \emph{asymptotically simple
  spacetimes} ---see \cite{CFEBook}, Section 10.2 for a definition
of this class of spacetimes. The Peeling theorem can be formulated as:

\begin{theorem}
\label{Theorem:Peeling}
Let $(\tilde{\mathcal{M}},\tilde{\bmg})$ denote a vacuum asymptotically simple
spacetime with vanishing Cosmological constant. Then the components of
the Weyl tensor with respect to a frame adapted to a foliation of
outgoing light cones satisfy
\[
\tilde\psi_0 = O\left( \frac{1}{\tilde{r}^5}\right),\quad \tilde\psi_1
= O\left( \frac{1}{\tilde{r}^4}\right),\quad \tilde\psi_2 = O\left(
  \frac{1}{\tilde{r}^3}\right),\quad \tilde\psi_3 = O\left(
  \frac{1}{\tilde{r}^2}\right),\quad \tilde\psi_4 = O\left( \frac{1}{\tilde{r}}\right),
\]
where $\tilde{r}$ is a suitable parameter along the generators of the
light cones.
\end{theorem}

The definition of asymptotically simple spacetimes involves an
assumption on the existence of a smooth (i.e. $C^\infty$) conformal
extension of the spacetime $(\tilde{\mathcal{M}},\tilde{\bmg})$. An inspection
of the proof of the Peeling theorem reveals that, in fact, it is only
necessary to assume that the conformal extension is $C^4$. In view of
the latter, the question of the existence and genericity of spacetimes
satisfying the \emph{peeling behaviour} can be rephrased in terms of
the construction of asymptotically flat spacetimes with, at least,
this minimum of differentiability.

\medskip
There exists a vast body of work aimed at the construction of
spacetimes satisfying the peeling behaviour and at understanding the
genericity of this property ---see e.g. \cite{Fra04,Fri99,CFEBook} for an entry
point to the literature on this subject. The seminal work in \cite{Pen65a} already
singles out a key feature of the problem ---namely, that whereas
\emph{Penrose's compactification procedure} applied to the Minkowski
spacetime renders a fully smooth conformal extension, for spacetimes
with a non-vanishing mass (e.g. the Schwarzschild spacetime) the
conformal structure degenerates at \emph{spatial infinity}. As spatial
infinity can be regarded as the (past/future) endpoint of the generators of
(future/past) null infinity, it is natural to expect that the behaviour of
the gravitational field near spatial infinity will, somehow, reflect
on the peeling properties of the spacetime ---particularly, if one
tries to analyse these from the point of view of a Cauchy initial
value problem.

\medskip
The first systematic attempt to understand the \emph{generic} properties of the
Einstein field equations near spatial infinity from the point of view
of an initial value problem are due to Beig \& Schmidt \cite{BeiSch82,Bei84} who
integrate the equations along the so-called \emph{hyperboloid} at
spatial infinity. Further insight on the relation between the peeling
property and spatial infinity was provided by Friedrich's proof of the
semiglobal existence and stability of perturbations of the Minkowski
spacetime from \emph{hyperboloidal initial data}
---see\cite{Fri86b}. This result ensures the existence of (semiglobal)
developments with a smooth null infinity (and thus peeling) if suitable hyperboloidal
initial data is provided. A posterior analysis of the solutions to the
Einstein constraint equations in the hyperboloidal setting by
Andersson, Chrusciel \& Friedrich \cite{AndChrFri92} and later
Andersson \& Chrusciel \cite{AndChr93,AndChr94} revealed that the
initial data sets considered in Friedrich's semiglobal results are
non-generic. More precisely, their results reveal the existence of
certain obstructions to the smoothness null infinity and suggest that
a consistent framework for the analysis of the asymptotics of the
gravitational field of isolated bodies is that of
\emph{polyhomogeneous expansions} ---i.e expansions involving powers
of $1/\tilde{r}$ and $\ln \tilde{r}$. The formal properties of
spacetimes possessing polyhomogeneous asymptotic expansions were explored 
in \cite{ChrMacSin95} ---similar types of expansions had been
considered earlier in e.g. \cite{NovGol82,Win85}; further properties
have been later discussed in \cite{Val98,Val99a,Val99b}. 

\medskip
The proof of the non-linear stability of the Minkowski spacetime by
Christodoulou \& Klainerman \cite{ChrKla93} provides a further body of
evidence of the non-generic character of spacetimes satisfying the
peeling behaviour. Indeed, their analysis provides spacetimes, which,
in the notation of Theorem \ref{Theorem:Peeling} satisfy:
\begin{equation}
\tilde\psi_0 = O\left( \frac{1}{\tilde{r}^{7/2}}\right),\quad \tilde\psi_1
= O\left( \frac{1}{\tilde{r}^{7/2}}\right),\quad \tilde\psi_2 = O\left(
  \frac{1}{\tilde{r}^3}\right),\quad \tilde\psi_3 = O\left(
  \frac{1}{\tilde{r}^2}\right),\quad \tilde\psi_4 = O\left(
  \frac{1}{\tilde{r}}\right),
\label{NonPeelingCK}
\end{equation}
---see \cite{Fri92}. Christodoulou \& Klainerman's original results
left open the question of whether these decay estimates are sharp or
could actually be improved through further analysis.  The relation
between non-peeling expansions and a \emph{no incoming radiation
condition} expressed in terms of the constancy of the Bondi mass on
$\mathscr{I}^-$ in Christodoulou \& Klainerman's construction has been
explored in \cite{Chr03} ---remarkably, it turns out that the no
incoming radiation condition is not enough
to preclude the development of logarithmic terms in the
expansions. Later, by making further assumptions on the initial data
Klainerman \& Nicol\`o have been able to ensure the existence of
spacetimes with the peeling property ---see
\cite{KlaNic03}\footnote{Notice, however, that the conditions on the
initial data considered in this reference exclude the Kerr
spacetime.}.

\subsubsection*{The cylinder at spatial infinity} 
Arguably, the most systematic approach to the analysis of the
properties of the Einstein field equation near spatial infinity is
based on Friedrich's construction of the \emph{cylinder at spatial
  infinity} \cite{Fri98a}. This framework makes use of the so-called \emph{extended conformal
Einstein field equations} and a gauge based on the properties of
conformal geodesics to formulate a regular initial value problem at
spatial infinity. Crucial in this framework is a careful
description of the singular behaviour of the Weyl tensor and the degeneracy of
the conformal field equations near spatial infinity ---both at the
level of data and solutions. In particular, the use of a gauge based on
conformal invariants supports the intuition that the singularities
identified in this construction are truly geometric ---as opposed to
singularities due to a deficiency of the gauge. This point is further
supported by the analysis of static and stationary solutions given in
\cite{Fri04,AceVal11}.

Friedrich's framework represents spatial infinity through an extended
set obtained first by blowing-up the point of spatial infinity to a
2-sphere and then by considering the timelike conformal geodesics
passing trough this sphere and orthogonal to some fiduciary initial
hypersurface $\mathcal{S}$. These curves rule the cylinder at spatial infinity
$\mathcal{I}$ and allow to transport information from past null
infinity to future null infinity. The cylinder $\mathcal{I}$
intersects null infinity transversally at so-called \emph{critical
  sets} $\mathcal{I}^\pm$ ---see Figure \ref{Figure:CylinderHorizontal}. 

\begin{figure}[t]
\centering
\includegraphics[width=0.9\textwidth]{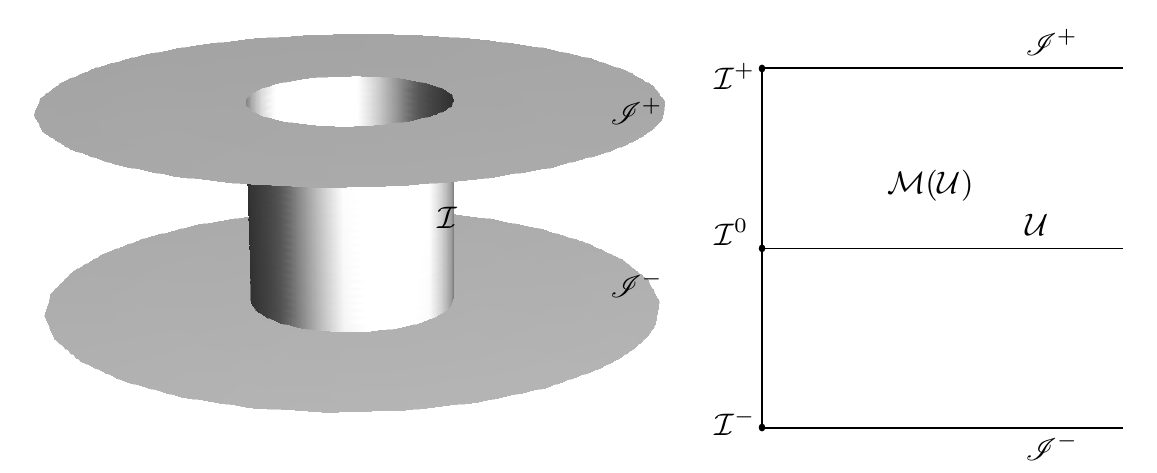}
\caption{Left: schematic representation of the cylinder at spatial
infinity in the so-called \emph{horizontal representation} where null
infinity corresponds to the locus of points with $\tau=\pm 1$. The
cylinder $\mathcal{I}$ is a total characteristic of the conformal
evolution equations ---see Section \ref{Section:ManifoldMa} for
further details. Right: longitudinal section in which the angular
dependence has been suppressed. Here $\mathcal{U}$ denotes an open set
in a neighbourhood of $i$ and $\mathcal{M}(\mathcal{U})$ its
development; $\mathcal{I}^\pm$ are the critical sets where the
cylinder meets spatial infinity and $\mathcal{I}^0$ is the
intersection of the cylinder with the initial hypersurface. These
figures are coordinate rather than conformal diagrams ---in
particular, conformal geodesics correspond to vertical lines.}
\label{Figure:CylinderHorizontal}
\end{figure}

From the point of view of the
conformal Einstein field equations the set $\mathcal{I}$ is very special as it
turns out to be a \emph{total characteristic} ---that is, the full
conformal evolution equations reduce to a system of transport
equations on the cylinder. Thus, it is not possible to prescribe
boundary data on $\mathcal{I}$. Rather, the value of the various
conformal fields is obtained from an integration of the transport
equations using the data at the intersection of the $\mathcal{I}$ with
$\mathcal{S}$. The integration of these transport equations reveals
that, in fact, the conformal evolution equations degenerate at the
critical sets ---in the sense that the matrix associated to the the
time derivatives in the evolution equations loses rank. 

Crucially, the total characteristic nature of the cylinder at spatial
infinity implies the existence of a hierarchy of transport equations
at $\mathcal{I}$. This, in turn, allows to
construct formal expansions of the solutions to the extended conformal
Einstein field equations ---so-called \emph{F-expansions}. These expansions provide valuable
information about the type of singular behaviour one can expect at the
\emph{intersection of null and spatial infinity}. In particular, in
\cite{Fri98a} it has been shown that for \emph{time symmetric initial data
sets}, the solutions to the hierarchy of transport equations at
$\mathcal{I}$ develop a particular type of logarithmic singularities
at the critical sets. This class of singularities 
are a manifestation of the degeneracy of the conformal
evolution equations at $\mathcal{I}^\pm$ and, roughly speaking, are
associated to the \emph{linear part} of the evolution equations. It is worth noticing that
the structural properties leading to this type of singular behaviour
is shared by a large class of equations ---the Maxwell, Yang-Mills,
scalar field, etc.

The key result in \cite{Fri98a} is that the class of logarithmic
singularities alluded to in the previous paragraph do not arise if the
\emph{regularity condition}
\begin{equation}
D_{\{i_p} D_{i_{p-1}} \cdots D_{i_1} b_{jk\}}(i) =0, \qquad
p=0,\,1,\,\ldots p
\label{FriedrichRegularityCondition}
\end{equation}
on the Bach tensor $b_{jk}$ of the \emph{initial conformal metric}
$\bmh$ and its derivatives at $i$ is satisfied. In condition
\eqref{FriedrichRegularityCondition} $D$ is the Levi-Civita connection
of $\bmh$ and ${}_{\{ab\cdots c\}}$ denotes the operation of taking
the symmetric, $\bmh$-tracefree part.  If condition
\eqref{FriedrichRegularityCondition} is not satisfied at some order
$p_\lightning$ then logarithmic singularities arise in the solutions
to the transport equations at order $p_\lightning+2$. The regularity
condition \eqref{FriedrichRegularityCondition} is a condition on the
conformal class $[\bmh]$ of the metric $\bmh$ and thus can be verified
to be conformally invariant. It is satisfied by static initial data
sets ---see \cite{Bei91b,Fri86a,Fri04}. It is also (trivially)
satisfied by conformally flat time symmetric initial data sets
---e.g. the Brill-Lindquist and Misner initial data sets
\cite{BriLin63,Mis63}. Another class of \emph{obstructions to the
smoothness of null infinity}, has been identified in
\cite{Val04a,Val04d}. The logarithmic singularities associated to
these obstructions are related to the specific nonlinear structure of
the Einstein field equations.

\medskip
As already emphasised, the asymptotic expansions obtained through the
framework of the cylinder at spatial infinity are obtained in a very
specific gauge (the \emph{F-gauge}) based on the properties of
conformal geodesics. Accordingly, a natural question to be asked is the
following:

\smallskip
\noindent
\emph{\textbf{Q}: what type of asymptotic expansions near null infinity is implied by
the formal F-expansions near the cylinder at spatial infinity?}

\smallskip
As the discussion of asymptotic properties near null infinity
(cf. e.g. Theorem \ref{Theorem:Peeling}) is
usually expressed in terms of a gauge hinged on null infinity (the
\emph{Newman-Penrose (NP) gauge}), in order to address \textbf{\em Q}
 one needs to relate the F and NP gauges. The required
analysis has been discussed in \cite{FriKan00}.

\subsubsection*{Main results}
In this article we address \textbf{\em Q} for time symmetric initial
data sets admitting a conformal metric that is analytic near spatial
infinity ---this is the class of initial data sets considered in
\cite{Fri98a}. For this class of initial data there exist detailed
computations of the solution to the transport equations near spatial
infinity ---see e.g. \cite{Val04a,Val04d,Val07a}. Under the assumption
\emph{that the formal expansions correspond to actual solutions to the
conformal Einstein field equations} (Assumption
\ref{Assumption:FormalExpansionsAreNotFormal} in Section
\ref{Section:CylinderSpatialInfinity}) our main result is that generic
time symmetric initial data (see Assumption \ref{Assumption:Initial
data} in Section \ref{Section:CylinderSpatialInfinity}) gives rise to
developments with expansions near null infinity of the form
\begin{equation}
\tilde\psi_0 = O\left( \frac{\ln \tilde{r}}{\tilde{r}^3}\right),\quad \tilde\psi_1
= O\left( \frac{\ln \tilde{r}}{\tilde{r}^3}\right),\quad \tilde\psi_2 = O\left(
  \frac{\ln \tilde{r} }{\tilde{r}^3}\right),\quad \tilde\psi_3 = O\left(
  \frac{1}{\tilde{r}^2}\right),\quad \tilde\psi_4 = O\left(
  \frac{1}{\tilde{r}}\right),
\label{NonPeelingExpansions}
\end{equation}
thus, suggesting \emph{polyhomogeneous asymptotics}. As 
\begin{equation}
\left( \frac{\ln \tilde{r}}{\tilde{r}^3}\right)\bigg/\left(
  \frac{1}{\tilde{r}^{7/2}} \right) = \tilde{r}^{1/2}\ln \tilde{r}
\longrightarrow \infty \qquad \mbox{as}\quad \tilde{r}\longrightarrow
\infty,
\label{UsVsCK}
\end{equation}
the expansions in \eqref{NonPeelingExpansions} imply a decay of the
components of the Weyl tensor which is
even slower than the one given by Christoduolou \& Klainerman
---cf. \eqref{NonPeelingCK}. Moreover, our analysis shows how the
non-peeling decay in \eqref{NonPeelingExpansions} is related to
Friedrich's regularity condition
\eqref{FriedrichRegularityCondition}. More precisely, the decay
\eqref{NonPeelingExpansions} is obtained if the regularity condition
\eqref{FriedrichRegularityCondition} is violated at order $p=0$. We
also show that, under our assumptions, a necessary condition to obtain
 peeling decay is that \eqref{FriedrichRegularityCondition} holds up
to order $p=3$. 

\medskip
Once the part of the initial data responsible for the decay in
\eqref{NonPeelingExpansions} has been identified, it is natural to ask
whether it is possible to construct a candidate (global) spacetime to
have such asymptotic behaviour. We show that this is indeed possible 
by first constructing a family of initial data sets which violate the
regularity condition \eqref{FriedrichRegularityCondition} to order
$p=0$ and which can be regarded as perturbations of time symmetric
Minkowski initial data. Consistent with the limit in
\eqref{UsVsCK},  it turns out that Christodoulou \&
Klainerman's result is not applicable to initial data sets that
violate Friedrich's regularity condition to order $p=0$ ---the
quantity measuring global smallness is not well defined. Fortunately, Bieri's
generalisation of Christodoulou \&
Klainerman's global existence result \cite{Bie10} can, nevertheless, be used
to construct the maximal global hyperbolic development of the initial
data. This construction render the desired candidate spacetime. 

\begin{remark}
{\em The various pieces of information required to obtain the
  expressions \eqref{NonPeelingExpansions} have been available in the
  literature for sometime. However, they have remained scattered as,
  for historical reasons, the main focus has for a while been to understand the
  conditions leading to a smooth conformal extension. The
  main task in this article is to put these pieces together.}
\end{remark}

\begin{remark}
{\em The analysis in this article can be generalised, to the expense
  of lengthier computations, to classes of initial with a non-vanishing
extrinsic curvature. As the decays in \eqref{NonPeelingExpansions} are
almost borderline with what is allowed by the Einstein field
equations, it is conjectured that the inclusion of a non-vanishing
extrinsic curvature without linear momentum will not modify our main
result.}
\end{remark}

\begin{remark}
{\em It should be stressed that our main result is formal ---i.e. it
  assumes the existence of a spacetime with the given asymptotics
  ---cf. Assumption \ref{Assumption:FormalExpansionsAreNotFormal}. The
main open problem in the subject of the asymptotics of the
gravitational field of isolated systems can be described as showing
that there indeed exist solutions to the Einstein field equations with
the given asymptotics ---in other words, one needs to analyse the
relation between the expansions used to obtain
\eqref{NonPeelingExpansions} and actual solutions. This is a challenging and deep problem which
deserves careful consideration. A possible line of attack of this
question is to obtain a suitable generalisation of the estimates for
the spin-2 equation in \cite{Fri03b} to the full conformal Einstein equations.}
\end{remark}

\subsection*{Outline of the article}
The article is structured as follows. Section \ref{Section:CFE}
provides a brief review of the main technical tool in this article,
the extended conformal Einstein field equations ---including a
discussion of gauge issues and the construction of initial
data. Section \ref{Section:CylinderSpatialInfinity} provides an
overview of the construction of the cylinder at spatial infinity. The
purpose of this section is not so much to describe the construction
but to set notation and describe the features relevant for the purposes of
the present article ---in particular, Subsection
\ref{Section:TransportEquations} provides a description of the
asymptotic expansions of the solutions to the conformal field
equations near the cylinder for time symmetric initial data. Section
\ref{Section:NPGauge} provides a discussion of the relation between
the F-gauge used in the construction of the cylinder at spatial
infinity and the NP-gauge used in the discussion of peeling
properties. Our main results are provided in Section
\ref{Section:PolyhomogeneousExpansions}. Section \ref{Section:Candidates} discusses the
construction of global spacetimes which are candidates to have the
non-peeling asymptotics of our main results. We provide some
concluding remarks and outlooks in Section
\ref{Section:Conclusions}. Finally, in the Appendix we provide
detailed expressions of the solutions to the transport equations on
the cylinder at spatial infinity. 

\subsection*{Notation and conventions}
In this article $\{_a ,_b , _c ,
 . . .\}$ denote abstract tensor indices and $\{_\bma ,_\bmb , _\bmc ,
 . . .\}$ will be used as spacetime frame indices taking the values ${
   0, . . . , 3 }$.  In this way, given a basis
$\{\bme_{\bma}\}$ a generic tensor is denoted by $T_{ab}$ while its
components in the given basis are denoted by $T_{\bma \bmb}\equiv
T_{ab}\bme_{\bma}{}^{a}\bme_{\bmb}{}^{b}$.
  Additionally, spatial frame indices respect to an
 adapted frame will be denoted by ${}_\bmi,\, {}_\bmj,\,
 {}_\bmk,\ldots$ and will take the values $1,\,2,\,3$.  Since part of
 the analysis will require the use of spinors the notation and
 conventions of Penrose \& Rindler \cite{PenRin84} will be followed.
 In particular, capital Latin indices $\{ _A , _B , _C , . . .\}$ will
 denote abstract spinor indices while boldface capital Latin indices
 $\{ _\bmA , _\bmB , _\bmC , . . .\}$ will denote frame spinorial
 indices with respect to some specified spin dyad ${
   \{\epsilon_\bmA{}^{A} \} }.$

\section{The conformal Einstein field equations}
\label{Section:CFE}
This section provides a brief discussion of the basic technical tools
used in this article. For a more detailed discussion on the structural
properties of the conformal
Einstein field equations, the reader is referred to \cite{CFEBook}. 

\subsection{The extended conformal Einstein field equations}
The \emph{extended conformal Einstein field equations} are a conformal
representation of the Einstein field equations expressed in terms of a
Weyl connection. In what follows let
$(\tilde{\mathcal{M}},\tilde\bmg)$ denote a spacetime satisfying the
vacuum Einstein field equations 
\begin{equation}
\tilde{R}_{ab}=0,
\label{EFE}
\end{equation}
and let $(\mathcal{M},\bmg)$ be a conformal extension thereof with
\[
g_{ab} = \Theta^2 \tilde{g}_{ab},
\]
where $\Theta$ is some suitable conformal factor. Furthermore, let
$\{ \bme_a \}$ denote a $\bmg$-orthonormal frame. In terms of the
above, the extended conformal Einstein field equations are given by
the concise expressions 
\begin{subequations}
\begin{eqnarray}
&& [\bme_\bma,\bme_\bma] = ( \hat\Gamma_\bma{}^\bmc{}_\bmb
   -\hat{\Gamma}_\bmb{}^\bmc{}_\bma  )\bme_\bmc, \label{XCFE1}\\
&& \hat{P}^\bmc{}_{\bmd\bma\bmb} = \hat{\rho}^\bmc{}_{\bmd\bma\bmb}, \label{XCFE2}
  \\
&& \hat{\nabla}_\bmc\hat{L}_{\bmd\bmb} -
   \hat{\nabla}_\bmd\hat{L}_{\bmc\bmb} = d_\bma
   d^\bma{}_{\bmb\bmc\bmd}, \label{XCFE3}\\
&& \hat{\nabla}_\bma d^\bma{}_{\bmb\bmc\bmd} = f_\bma
   d^\bma{}_{\bmb\bmc\bmd}. \label{XCFE4}
\end{eqnarray}
\end{subequations}
In the previous equations $\hat{\nabla}$ denotes the Weyl connection
defined through the relation
\[
\hat{\nabla}_a g_{bc} = -2 f_a g_{bc},
\]
with $f_a$ a smooth covector and $\hat\Gamma_\bma{}^\bmb{}_\bmc$ the
associated connection coefficients. Moreover, 
\[
d_\bma = \Theta f_\bma + \nabla_\bma \Theta,
\]
while
\[
d^\bma{}_{\bmb\bmc\bmd} \equiv \Theta^{-1} C^\bma{}_{\bmb\bmc\bmd}
\]
are the components of the \emph{rescaled Weyl tensor} and
$\hat{L}_{\bma\bmb}$ denote components of the Schouten
tensor of the connection $\hat{\nabla}$. Finally, $\hat{P}^\bmc{}_{\bmd\bma\bmb}$ and
$\hat{\rho}^\bmc{}_{\bmd\bma\bmb}$ denote, respectively, the
\emph{geometric} and \emph{algebraic curvatures} ---i.e. the
classical expression of the curvature in terms of the connection
coefficients and its irreducible decomposition in terms of
$d^\bma{}_{\bmb\bmc\bmd}$ and $\hat{L}_{\bma\bmb}$. 

\begin{remark}
{\em Equations \eqref{XCFE1} and \eqref{XCFE2} are a rewriting of the
first and second Cartan structure equations. In particular, equation
\eqref{XCFE2} provides a differential condition for the components $f_\bma$.}
\end{remark}

\begin{remark}
{\em A key property of the system \eqref{XCFE1}-\eqref{XCFE4} is that
  a solution thereof implies, whenever $\Theta\neq 0$, a solution to
  the Einstein field equations \eqref{EFE} ---see e.g. \cite{CFEBook},
  Proposition 8.3. }
\end{remark}

\begin{remark}
{\em Most applications of the extended conformal Einstein field
  equations make use of a spinorial frame version of thereof. The latter
  is readily obtained by contraction with the constant Infeld-van der
Waerden symbols $\sigma^\bma{}_{\bmA\bmA'}$. In particular, exploiting
the symmetries of the rescaled Weyl tensor one has that its spinorial
counterpart is defined as 
\[
d_{\bmA\bmA'\bmB\bmB'\bmC\bmC'\bmD\bmD'} \equiv
\sigma^\bma{}_{\bmA\bmA'}\sigma^\bmb{}_{\bmB\bmB'}\sigma^\bmc{}_{\bmC\bmC'}\sigma^\bmd{}_{\bmD\bmD'} d_{\bma\bmb\bmc\bmd},
\]
so that one has the decomposition
\[
d_{\bmA\bmA'\bmB\bmB'\bmC\bmC'\bmD\bmD'} = -\phi_{\bmA\bmB\bmC\bmD}
\epsilon_{\bmA'\bmB'} \epsilon_{\bmC\bmC'} -
\bar{\phi}_{\bmA'\bmB'\bmC'\bmD'} \epsilon_{\bmA\bmB}\epsilon_{\bmC\bmD},
\]
with $\phi_{\bmA\bmB\bmC\bmD}=\phi_{(\bmA\bmB\bmC\bmD)}$ the
components of the \emph{Weyl spinor} $\phi_{ABCD}$. Equation
\eqref{XCFE4} then takes the form
\[
\hat{\nabla}^\bmA{}_{\bmA'} \phi_{\bmA\bmB\bmC\bmD} = f^\bmA{}_{\bmA'} \phi_{\bmA\bmB\bmC\bmD} .
\]
}
\end{remark}

\subsection{Gauge considerations}
A conformal geodesic on $(\tilde{\mathcal{M}},\tilde\bmg)$ can be
viewed as a curve $x(\tau)$, $\tau\in I$ with $I$ some interval, for
which there exists a Weyl connection along the curve $\hat{\nabla}$ such that 
\[
\hat{\nabla}_{\dot{\bmx}} \dot{\bmx}=0,
\]
with $\dot{\bmx}$ the tangent vector to the curve ---in other words,
$x(\tau)$ is an affine geodesic of the curve, see \cite{CFEBook}
Section 5.5.2. This property can be used to construct \emph{Gaussian
gauge systems} if a given region of spacetime is covered by a
non-intersecting congruence of conformal geodesics. In this type of
gauge systems an orthonormal frame $\{ \bme_\bma \}$ is Weyl
propagated and spatial coordinates $\underline{x}_\star$ are extended
off an initial hypersurface by requiring them to be constant along the
curve. These \emph{spatial} coordinates are completed by including the
parameter $\tau$ of the curve so as to obtain the spacetime
coordinates $(\tau,\underline{x}_\star)$. Moreover, conformal geodesics
single out a \emph{canonical representative} of the conformal class
$[\tilde\bmg]$ by imposing the condition
\begin{equation}
\bmg(\dot{\bmx},\dot{\bmx}) = \Theta^2 \tilde\bmg(\dot\bmx,\dot\bmx)
=1.
\label{NormalisationCondition}
\end{equation}
Remarkably, the above condition leads to an explicit expression for
the conformal factor $\Theta$ in terms of the parameter of the curve
---namely
\begin{equation}
\Theta(\tau) = \Theta_\star + \dot\Theta_\star (\tau-\tau_\star) +
\frac{1}{2}\ddot\Theta_\star (\tau-\tau_\star)^2, 
\label{ConformalFactorCG}
\end{equation}
where the coefficients $\Theta_\star$, $\dot\Theta_\star$ and
$\ddot\Theta_\star$ are constant along a given conformal geodesic and
are expressible in terms of initial data for the congruence ---see
\cite{CFEBook}, Proposition 5.1. In view of
\eqref{NormalisationCondition} one can set
$\bme_\bmzero=\dot{\bmx}$. It follows then that
\[
\hat{\Gamma}_\bmzero{}^\bma{}_\bmb=0, \qquad f_\bmzero =0, \qquad \hat{L}_{\bmzero\bma}=0.
\]

\medskip
The extended conformal Einstein field equations
\eqref{XCFE1}-\eqref{XCFE4} expressed in terms of a conformal Gaussian
gauge system as described in the previous paragraph take the form
\begin{subequations}
\begin{eqnarray}
&& \partial_\tau \bme = \mathbf{Q}_1(\bmGamma,\bme) + \mathbf{K}_1(\bme), \label{CEE1}\\
&& \partial_\tau\bmGamma =
   \mathbf{Q}_2(\bmGamma,\bmGamma) + \mathbf{K}_2(\bmPhi)+\mathbf{L}_2\bmphi, \label{CEE2}\\
&& \partial_\tau \bmPhi = \mathbf{Q}_3(\bmGamma,\bmPhi )
   + \mathbf{L}_3\bmphi, \label{CEE3}\\
&& \big( \mathbf{I} + \mathbf{A}^0(\bme)  \big)\partial_\tau \bmphi +
   \mathbf{A}^\alpha \partial_\alpha \bmphi
   =\mathbf{B}(\bmGamma,\bmphi), \label{CEE4}
\end{eqnarray}
\end{subequations}
where $\bme$, $\bmGamma$, $\bmPhi$ and $\phi$ denote the independent components of the frame,
connection, Schouten tensor and Weyl tensor. In equations
\eqref{CEE1}-\eqref{CEE3} one has that 
$\mathbf{Q}_1$, $\mathbf{Q}_2$ and $\mathbf{Q}_3$ denote a quadratic
expression with constant coefficients of the arguments, $\mathbf{K}_1$
and $\mathbf{K}_2$ denote linear functions with constant coefficients
of the argument and $\mathbf{L}_2(x)$ and $\mathbf{L}_2(x)$ denote smooth matrix-values
functions of the coordinates. Moreover, in equation \eqref{CEE4},
$\mathbf{I}$ is the $5\times 5$ unit matrix, $\mathbf{A}^\mu$ are
  Hermitian matrices depending smoothly on the frame coefficients
  $\bme$ and $\mathbf{B}$ is a quadratic expression with constant
  coefficients of its arguments. 

\subsection{Initial data for the evolution equations}
Initial data for the conformal evolution equations
\eqref{CEE1}-\eqref{CEE4} can be constructed if a conformal metric
$\bmh$ and a conformal factor $\Omega$ are provided such that 
\[
r[\Omega^{-2} \bmh]=0
\]
so that the time symmetric Hamiltonian constraint is satisfied. In
particular, for the components $d_{\bmi\bmj}$ of the electric part of
the Weyl tensor (the magnetic part vanishes due to the time symmetry)
one has that
\begin{equation}
d_{\bmi\bmj} =\frac{1}{\Omega^2}\big( D_{\{\bmi} D_{\bmj\}} \Omega +
\Omega s_{\bmi\bmj}  \big)
\label{InitialData:ElectricPartWeyl}
\end{equation}
with $s_{\bmi\bmj}$ the components of the tracefree Ricci tensor of the metric
$\bmh$. Further details on the construction of initial data for the
conformal evolution equations can be found in \cite{CFEBook}, Section
11.4.3. 

\medskip
\noindent
\textbf{The Bach tensor.} The Bach tensor of the conformal metric
$\bmh$ will play key a role in the subsequent discussion. Given the
Schouten tensor $l_{ij}$ of $\bmh$, the \emph{Cotton tensor} is given by
\[
b_{kij} \equiv D_i l_{jk} -D_j l_{ik}.
\]
Its Hodge dual $b_{ij}$, is called the \emph{Bach
  tensor}, and is given by
\[
b_{ij} \equiv \frac{1}{2}\epsilon_j{}^{kl} b_{ikl},
\]
so that
\[
b_i{}^i=0, \qquad b_{ij}=b_{ji}, \qquad D^ib_{ij}=0.
\]
The spinorial counterpart of $b_{ij}$ is a totally symmetric
spinor $b_{ABCD}$ ---the \emph{Bach spinor}. Finally it is observed that under the rescaling
$\bmh \mapsto \bmh'=\phi^2 \bmh$ one has that
\[
b'_{ijk} = b_{ijk}, \qquad b'_{ij} =\phi^{-1}b_{ij}, \qquad b'_{ABCD}
= \phi^{-1} b_{ABCD}.
\]

\section{The cylinder at spatial infinity}
\label{Section:CylinderSpatialInfinity}
In this section we provide a brief discussion of the construction of
the cylinder at spatial infinity and the asymptotic expansions (\emph{F-expansions}) it
gives rise to. The
basic references on the construction of the cylinder at spatial infinity are
\cite{Fri98a,Fri04} ---see also Chapter 21 in \cite{CFEBook}.

\subsection{Boundary conditions}
In the present discussion we will restrict our attention to time symmetric initial
data sets for the Einstein field equations which are asymptotically
Euclidean and admit a point compactification yielding an analytic
conformal metric. More precisely, we assume that there exists a
3-dimensional Riemannian manifold $(\mathcal{S},\bmh)$ with a point
$i\in\mathcal{S}$ and a function $\Omega\in C^2$ such that
\begin{equation}
\Omega(i)=0, \qquad \mathbf{d}\Omega(i) =0, \qquad \mathbf{Hess}\,
\Omega(i) =-2 \bmh(i),
\label{BoundaryConditions}
\end{equation}
with $\Omega>0$ away from $i$ and $\bmh$ analytic at least in a
neighbourhood of $i$ ---cf. \cite{CFEBook}, Definition 11.2.

\begin{remark}
{\em In the following, our discussion will the restricted to a suitably
small neighbourhood of $i$ on $\mathcal{S}$ and the Cauchy development thereof.}
\end{remark}

 \begin{remark}
 {\em There is some conformal gauge freedom left in Conditions \eqref{BoundaryConditions}. A replacement of the form $\bmh\mapsto\vartheta^4\bmh$, $\Omega \mapsto \vartheta^2\Omega$ with
 $\vartheta(i)=1$ gives rise to the same physical metric $\tilde\bmh
   =\Omega^{-4}\bmh$. This gauge freedom has been used in \cite{Fri98a}
 to construct a \emph{conformal normal gauge} for which there exist
 coordinates $\underline{x} =(x^\alpha)$ with $x^\alpha(i)=0$ such
 that 
\begin{equation}
h_{\alpha\beta} =-\delta_{\alpha\beta} + O(|x|^3).
\label{CNGauge}
\end{equation}
In particular, the curvature of $\bmh$ vanishes at $i$.
 }
\end{remark}

In what follows, we make the following assumption:

\begin{assumption}
\label{Assumption:Initial data}
The metric $\bmh$ satisfies the
boundary conditions \eqref{BoundaryConditions} with a conformal
factor $\Omega\in C^2(\mathcal{S})\cap C^\infty(\mathcal{S}\setminus
\{ i\})$. Moreover, it is analytic in a
neighbourhood of $i$ and there exists coordinates
$\underline{x}=(x^\alpha)$ for which the components of $\bmh$ satisfy \eqref{CNGauge}.
\end{assumption}

\begin{remark}
{\em The assumption of analyticity has been made for convenience and
  easy reference with the analysis in \cite{Fri98a}. As pointed out in
  that reference, this assumption is not essential and for particular
  computations (like the ones considered here) a finite degree of
  differentiability suffices ---see also similar remarks in \cite{Fri13}.} 
\end{remark}

\medskip
Under Assumption \ref{Assumption:Initial data} it follows that the
conformal factor $\Omega$ admits, in a suitable neighbourhood
$\mathcal{B}_a(i)$ of $i$,
the parametrisation
\[
\Omega = \frac{U}{|x|} + W
\]
where $U$ and $W$ are analytic functions on $\mathcal{B}_a(i)$ 
\[
U=1 + O(|x|^4), \qquad W(i) =\frac{m}{2},
\]
where $m$ is the ADM mass of the initial data. The function $U/|x|$ is
the Green's function of the Yamabe equation implied by the time
symmetric Hamiltonian constraint and contains information about the
local geometry around $i$ while $W$ encodes global information ---in
particular the mass. The function $W$ can be expanded as
\[
W = \frac{m}{2}+ W_1 |x| + \frac{1}{2} W_2 |x|^2 +O(|x|^3),
\]
where $W_1$, $W_2$ are smooth functions of the \emph{angular coordinates}.

\medskip
\noindent
\textbf{Singular behaviour of the rescaled Weyl tensor.} A direct
computation combining the expansions in the previous paragraphs with
equation \eqref{InitialData:ElectricPartWeyl} leads to
\[
d_{\bmi\bmj}= -\frac{3m x_{\{\bmi}x_{\bmj\}}}{|x|^5} + O(|x|^{-2}),
\]
where $x_{\bmi} \equiv \delta_\bmi{}^\beta \delta_{\alpha\beta}
x^\alpha$ is the position vector centred at $i$. Consequently, one has 
 that the Weyl tensor blows up at $i$ as $|x|^{-3}$. 

\begin{remark}
{\em As pointed out in the introduction the above singularity is the
  main technical difficulty in formulating an initial value problem
  for the conformal Einstein field equations in a neighbourhood of
  spatial infinity.}
\end{remark}

\subsection{The manifold $\mathcal{C}_{a}$}
As it is well known, the standard approach to the discussion of the
conformal structure of asymptotically flat spacetimes makes use of
representations in which spatial infinity is a point $i^0$. By
contrast, the representation introduced in \cite{Fri98a} describes
spatial infinity in terms of an extended set ---the \emph{cylinder at
  spatial infinity}. Friedrich's representation makes use of a blow up
of the point at infinity, $i$, on the initial hypersurface
$\mathcal{S}$ to a 2-sphere. The construction makes use of a
particular bundle of spin-frames (\emph{bundle space}) over $\mathcal{B}_a(i)$. As we will
not require the full details of this construction, in the following we
provide a brief description in terms of (vector) frames.

\medskip
\noindent
\textbf{The bundle space.} Given an $\bmh$-orthonormal frame $\{
\bme_\bmi\}$ at $i$, any other
frame at $i$ can be obtained by means of a rotation ---that is, any
other frame is of the form $e_\bmi(s) =s^\bmj{}_\bmi \bme_\bmj$ with
$s=(s^\bmj{}_\bmi)\in SO(3)$. In particular $\bme_\bmthree(s)$ sweeps all
possible directions at $i$ as one lets $s$ exhaust $SO(3)$. For a
given value of $s$, one distinguishes $\bme_\bmthree(s)$ as the
radial vector at $i$. Keeping $s$ fixed, one then constructs the
$\bmh$-geodesic starting $i$ that has tangent vector $\bme_3(s)$ and
denote by $\rho$ the affine parameter along the curve that vanishes at
$i$. One then parallely propagates the rest of the frame $\{\bme_\bmi\}$
along this curve. For a particular value of the parameter $\rho$, let denote
by $\{\bme_\bmi(\rho,s)\}$ the frame thus obtained. The subsequent
discussion will be then restricted to a suitable small metric ball
$\mathcal{B}_a(i)$ on which this construction can be carried out.  Let
$SO(\mathcal{B}_a(i))$ denote the bundle of oriented orthonormal
 frames over $\mathcal{B}_a(i)$. It follows then that the map from the
 set $(-a,a)\times SO(3)$ into $SO(3)$ given by $(\rho,s)\mapsto \bme_\bmi
 (\rho,s)$ as described above defines a smooth embedding of a
4-dimensional manifold into $SO(\mathcal{B}_a(i))$. In the following,
only non-negative values of $\rho$ will be considered.

\medskip
\noindent
\textbf{The blow-up of $i$.} Denote by $\mathcal{C}_a$ the image of the set $[0,a)\times
SO(3)$ and define
\[
\mathcal{I}^0 \equiv \{ (\rho,s)\in \mathcal{C}_a \; | \; \rho=0 \}
\approx SO(3). 
\]
Finally, denote denote by $\pi$ the projection of
$SO(\mathcal{C}_a(i))$ onto $\mathcal{C}_a(i)$. One then restricts the
attention to the the subgroup $SO(2)$ of $SO(3)$ which leaves
$\bme_\bmthree$ invariant ---i.e. $SO(2) \equiv \{ s' \in SO(3)\;|\;
s^{\prime \bmj}{}_{\bmthree} \bme_\bmj =\bme_\bmthree
\}$. Accordingly, if $s\in SO(3)$ and $s'\in SO(2)$ then
$\bme_\bmi(s)$ and $\bme_\bmi (s s')$ are parellely transported along
the same geodesic. It then follows that $\pi (\bme_\bmi (\rho,s))= \pi
(\bme_\bmi (\rho, s s')),$ so that the projection factorises as 
\[
\mathcal{C}_a \stackrel{\pi'}{\longrightarrow} \mathcal{C}'_a \equiv
\mathcal{C}_a/SO(2) \stackrel{\pi''}{\longrightarrow} \mathcal{B}_a(i).
\]
It can be then verified that the projection $\pi''$ maps
$\pi'(\mathcal{I}^0)\approx \mathbb{S}^2$ onto $i$ and induces an
isomorphism between $\mathcal{C}'_a\setminus \pi'(\mathcal{I}^0)$ onto
the punctured ball $\mathcal{B}_a(i)\setminus \{ i\}$. This
diffeomorphism allows to identify these sets. Notwithstanding the
above factorisation, in \cite{Fri98a} the data for the conformal evolution
equations, e.g. formula \eqref{InitialData:ElectricPartWeyl}, on
$\mathcal{B}_a(i)$ is lifted to
$\mathcal{C}_a$  using $\pi$. In this manner $\mathcal{C}_a$ becomes
the initial manifold, $(\rho,s)$ are used as coordinates and the
boundary $\mathcal{I}^0$ becomes a blow-up of $i$. The manifold
$\mathcal{C}_a$ has one dimension more than $\mathcal{B}_a(i)$ ---this
has to do with the action of $SO(2)$. 

\medskip
\noindent
\textbf{A frame formalism.} It can be verified that all the
 fields appearing in the extended conformal Einstein field equations
 have a well-defined transformation behaviour (spin-weight) under such
 action. Thus, it is possible to define vector fields
 $\{\bmX,\,\bmc_\bmi(\rho,s)\}$ such that $\bmX$ is generated by the action
 of $SO(2)$ and the vector fields $\{ \bmc_\bmi(\rho,s)\}$ project to
 $\bme_\bma(\rho,s)$ via $\pi$. These vector fields allow the
 introduction of a frame formalism on $\mathcal{C}_a$ in which the
 frame is expanded by $\bmpartial_\rho$ and vectors $\bmX_\pm$ which
 are a basis of the Lie algebra $\mathfrak{s}\mathfrak{u}(2)$ ---the
 vectors $\bmX_\pm$ can be shown to be related to the $\eth$ and
 $\bar\eth$ operators, see \cite{FriKan00}. 

\medskip
\noindent
\textbf{Fixing the conformal gauge.} The initial data for the
conformal factor in the conformal Gaussian gauge system is set by
requiring
\[
\Theta_\star \equiv \kappa^{-1}\Omega,
\]
where the function $\kappa$ is chosen so that
\[
\kappa =\rho \kappa', \qquad \kappa'\in C^\infty(\mathcal{C}_a),
\qquad \kappa'>0, \qquad \bmX \kappa'=0,\qquad \kappa'|_{\mathcal{I}^0}=1.
\]
The function $\kappa$ induces a conformal rescaling of the frame
$\bme_\bma \mapsto \kappa\bme_a$. This rescaling maps bijectively the
set $\mathcal{C}_{a,\kappa}\setminus \mathcal{I}^0$ onto a smooth
submanifold of the bundle of frame fields over $\mathcal{B}_a$. This
submanifold is denoted by $\mathcal{C}_{a,\kappa}$. In order to
complete the construction of the conformal Gaussian system, one needs
to provide initial data for the 1-form $f_a$. This is done setting
\[
f_\bmzero =0, \qquad f_\bmi = \kappa^{-1} D_\bmi \kappa.
\]

\subsection{The manifold $\mathcal{M}_{a,\kappa}$}
\label{Section:ManifoldMa}
Following the previous discussion, the conformal evolution
  equations \eqref{CEE1}-\eqref{CEE4}  can be regarded as equations on
  the development of $\mathcal{C}_{a,\kappa}$. This development will
  be denoted as $\mathcal{M}_{a,\kappa}$ and is a 5-dimensional
  manifold embedded in the bundle of frame fields over the unphysical
  spacetime $\mathcal{M}$. The manifold $\mathcal{M}_{a,\kappa}$ is a
  $SO(2)$-bundle over the spacetime. The projection sending
  $\mathcal{M}_{a,\kappa}$ to $\mathcal{M}$ will be denoted, in a
  slight abuse of notation, by $\pi$. The coordinates $(\rho,s)$ and
  the vector fields $\bmX$, $\bmc_\bmi$ are extended from
  $\mathcal{C}_{a,\kappa}$ into $\mathcal{M}_{a,\kappa}$ by the flow
  of conformal geodesics ruling $\mathcal{M}_{a,\kappa}$ in such a way
  that the vectors do not pick up a component in the direction of the
  fibres ---i.e. in the direction of $\bmX$. Further, we make use of
  the parameter $\tau$ of the conformal geodesics as a further
  coordinate on $\mathcal{M}_{a,\kappa}$ ---that is, we set
  $x^0=\tau$. 

\medskip
From the conformal gauge choice it follows that the conformal factor
\eqref{ConformalFactorCG} takes the form
\[
\Theta = \Theta_\star \bigg( 1-
\frac{\kappa_\star^2}{\omega_\star^2}\tau^2  \bigg)\qquad \mbox{on}
\qquad \mathcal{M}_{a,\kappa}
\]
where 
\[
\omega \equiv \frac{2\Omega}{\sqrt{|D_i \Omega D^i \Omega|}}.
\]
The subscript ${}_\star$ indicates that the relevant functions are
constant along a given conformal geodesic. 

\begin{remark}
{\em A key property in the previous construction is that if the
  initial data for the conformal evolution equations has a smooth
  limit as $\rho\rightarrow 0$, then it can be smoothly extended into
  the coordinate range $\rho\leq 0$. Similarly for the fields $\Theta$
and $d_a$. Thus, the initial value problem for the conformal evolution
equations can be extended smoothly into a range where $\rho\leq 0$ in
such a way that the reduced equations are still symmetric hyperbolic.}
\end{remark}

Assuming that the development of the initial value problem for the
conformal evolution equations \eqref{CEE1}-\eqref{CEE4} extends far
enough then one can write
\begin{eqnarray*}
&& \mathcal{M}_{a,\kappa} = \left\{ |\tau|\leq \frac{\omega}{\kappa}, \;
  s\in SO(3)  \right\}, \\
&& \tilde{\mathcal{M}}_{a,\kappa} = \left\{ |\tau|< \frac{\omega}{\kappa}, \;
  s\in SO(3)  \right\}.
\end{eqnarray*}
In addition, it is natural to define
\begin{eqnarray*}
&& \mathscr{I}^\pm \equiv \left\{ \tau =\pm \frac{\omega}{\kappa}, \; s \in   SO(3)\right\}, \\
&& \mathcal{I} \equiv \left\{ |\tau|<1, \; \rho=0, \; s\in SO(3) \right\}
\end{eqnarray*}
 and
\begin{eqnarray*}
 && \mathcal{I}^\pm \equiv \left\{ \tau = \pm 1,\; \rho=0, \; s\in    SO(3) \right\},\\
&& \mathcal{I}^0 \equiv \left\{ \tau =0, \; \rho=0, \; s\in SO(3) \right\}.
\end{eqnarray*}

Observing that $\kappa/\rho \rightarrow 1$ as $\rho\rightarrow 0$ it
 can be readily checked that
 \begin{eqnarray*}
 && \Theta >0 \qquad \mbox{on} \quad \tilde{\mathcal{M}}_{a,\kappa}, \\
 && \Theta=0, \quad \mathbf{d}\Theta \neq 0 \qquad \mbox{on} \quad
     \mathscr{I}^-\cup \mathscr{I}^+\cup \mathcal{I}, \\
  && \Theta=0, \quad \mathbf{d}\Theta =0 \qquad \mbox{on} \quad \mathcal{I}^+\cup\mathcal{I}^-,
\end{eqnarray*}
 justifying the names of the various pieces of the conformal
 boundary. Moreover, the initial hypersurface corresponds to 
 \[
 \mathcal{C}_{a,\kappa}\equiv \left\{ \tau =0, \; 0\leq \rho<a, \;
   s\in SO(3)  \right\}, 
 \]
 so that its boundary is given by $\mathcal{I}^0$. 

\begin{remark}
{\em By assumption, the initial data extend smoothly, and in a unique
  manner to $\mathcal{I}^0$. The solution on
  $\mathcal{M}_{a,\kappa}$ depends only on the data given in
  $\mathcal{C}_{a,\kappa}$ as the set $\mathcal{I}$ is a total
  characteristic of the of the conformal evolution equations
  \eqref{CEE1}-\eqref{CEE4} so that the solution on $\mathcal{I}$
  depends only on the data on $\mathcal{I}^0$ ---that is, no boundary
  data can be prescribed. The set $\mathcal{I}$ is known as the
  \emph{cylinder at spatial infinity}. It can be regarded as a blow-up of $i^0$. Of
  particular relevance in the following discussion are the sets
  $\mathcal{I}^\pm$, the \emph{critical sets}, where the conformal
  evolution equations degenerate. Dealing with the consequences of this
degeneracy is the key issue of the \emph{problem of spatial
  infinity}. }
\end{remark}

\subsection{The solutions to the transport equations}
\label{Section:TransportEquations}
In \cite{Fri98a} it has been shown how the construction of the
cylinder at spatial infinity can be used to construct Taylor like
expansions of the solutions to the conformal evolution equations
\eqref{CEE1}-\eqref{CEE4} ---to be called \emph{F-expansions}. In order to do this, one differentiates the
conformal evolutions with respect to the radial direction $\rho$ an
arbitrary number of times, say $p$, and then evaluates at the
$\rho$. In what follows, for conciseness we restrict our attention to
the Bianchi equations \eqref{CEE4}. In particular, for the Bianchi
evolution equations \eqref{CEE4}, letting
\[
\bmphi^{[p]} \equiv \partial_\rho^p \bmphi\big|_{\rho=0}, \qquad p=0,\,1,\,2,\ldots,
\]
a calculation yields that 
\begin{eqnarray}
&& \big( \mathbf{I} + \mathbf{A}^0( \bme^{[0]})\big)\partial_\tau
\bmphi^{[p]} + \mathbf{A}^+(\bme^{[0]}) \bmX_+ \bmphi^{[p]} +
\mathbf{A}^-(\bme^{[0]}) \bmX_- \bmphi^{[p]} \nonumber\\
&& \hspace{1cm}= \mathbf{B}(\hat{\Gamma}^{[0]},\bmphi^{[p]}) 
+ \sum_{j=1}^p \binom{p}{j} \big(
\mathbf{B}(\hat{\Gamma}^{[j]},\bmphi^{[p-j]}) -
\mathbf{A}^+(\bme^{[j]}) \bmpartial_+ \bmphi^{[p-j]} -
   \mathbf{A}^-(\bme^{[j]}) \bmpartial_- \bmphi^{[p-j]}
   \big). \nonumber \\
&& \label{BianchiTransportEquation}
\end{eqnarray}
The above equation will be called the \emph{Bianchi transport equation
  of oder $p$}. Similar equations can be obtained for
\[
\bme^{[p]} \equiv \partial_\rho^p \bme\big|_{\rho=0}, \quad
\bmGamma^{[p]} \equiv \partial_\rho^p \bmGamma\big|_{\rho=0}, \quad \bmPhi^{[p]} \equiv \partial_\rho^p \bmPhi\big|_{\rho=0} \qquad p=0,\,1,\,2,\ldots.
\]

\begin{remark}
{\em Observe the absence of $\rho$-derivatives in equation
  \eqref{BianchiTransportEquation}. This is a consequence of the fact
  that $\mathcal{I}$ is a total characteristic of the conformal
  evolution equations \eqref{CEE1}-\eqref{CEE4}. A similar property is
  satisfied by the equations for $\bme^{[p]}$, $\bmGamma^{[p]}$,
  $\bmPhi^{[p]}$. Collecting all these transport equations for
  $p=0,\,1,\,2,\ldots p$ one obtains a hierarchy of equations which
  can be solved recursively starting from $p=0$. This is a process completely amenable to an
implementation in a computer algebra system ---see
\cite{Val04a,Val04d}.
Thus, equation
  \eqref{BianchiTransportEquation} constitutes a system of linear
  evolution equations intrinsic to $\mathcal{I}$ for $\bmphi^{[p]}$.}
\end{remark}

 If the lower terms 
\[
\bme^{[j]}, \qquad \bmGamma^{[j]}, \quad j=0,\,1,\ldots p
\]
and
\[
\bmphi^{[j]}, \quad  j=0,\,1,\ldots p-1,
\]
are known then one can use the restriction of the initial data at $\mathcal{I}$,
$\bmphi^{[p]}_\star =\bmphi^{[p]}(0,s)$ to solve for
$\bmphi^{[p]}=\bmphi^{[p]}(\tau,s)$. The result can be collected in an expansion
\[
\bmphi \simeq \sum^p_{j=0} \frac{1}{j!}\bmphi^{[j]} \rho^j, 
\]
where the symbol $\simeq$ has been used to indicate the formal character
of the expansion. The relation of these expansion to actual
solutions to the conformal Einstein field equations is a challenging
open problem.

\subsubsection{Solutions in the $\kappa=\omega$ gauge}
The solutions of the transport equations to
\eqref{BianchiTransportEquation} for the orders $p=0,\;1,\;2$ with
initial data satisfying Assumption \ref{Assumption:Initial data} have
been studied in \cite{Fri98a} in a conformal gauge for which
$\kappa=\rho$. Solutions for the orders $p=3,\; 4,\; 5$ have been
computed in \cite{Val04d}. The conformal gauge for which $\kappa=\rho$
was used in these references as it renders simpler expressions. In the
present article we are interested in analysing expansions near null
infinity. For this, a conformal gauge for which 
\[
\kappa=\omega
\]
 is
more convenient as the location of future null infinity is given by
the simple condition $\tau=1$. Calculations in this gauge have been
reported in  \cite{Val04d} ---see also \cite{Val07a}. In the rest of
this section
we provide a summary of the properties of these expansions which will
be required in the remainder of the article.

\begin{proposition}
\label{Proposition:ExpansionsCylinder}
For initial data satisfying Assumption \ref{Assumption:Initial data},
the solutions to the transport equation for the components of the
rescaled Weyl tensor, equation \eqref{BianchiTransportEquation},
satisfy:
\begin{itemize}
\item[(i)] At orders $p=0$ and $p=1$ the solutions have polynomial
  dependence in $\tau$.

\item[(ii)] At order $p=2$ the solutions have polynomial dependence in
  $\tau$ if and only if the condition
\[
b_{ij}(i) =0
\]
is satisfied by the Bach tensor of the conformal metric $\bmh$. If
$b_{ij}(i) \neq 0$ the solution develops logarithmic singularities at
$\tau=\pm 1$. 

\item[(iii)] At order $p=3$ the solutions have polynomial dependence if
  and only if
\[
b_{ij}(i) =0, \qquad D_{\{ i} b_{jk\}}(i)=0.
\]
If $D_{\{ i} b_{jk\}}(i)\neq0$ then the solution develops logarithmic singularities at
$\tau=\pm 1$.

\item[(iv)] At order $p=4$ the solutions have polynomial dependence if
  and only if
\[
b_{ij}(i) =0, \qquad D_{\{ i} b_{jk\}}(i)=0, \qquad D_{\{ i} D_j b_{kl\}}(i)=0. 
\]
If $D_{\{ i} D_j b_{kl\}}(i)\neq0$ then the solution develops logarithmic singularities at
$\tau=\pm 1$.

\item[(v)] At order $p=5$ if
\[
b_{ij}(i) =0, \qquad D_{\{ i} b_{jk\}}(i)=0, \qquad D_{\{i_2} D_{i_1}
b_{jk\}}(i)=0, \qquad D_{\{i_3} D_{i_2}D_{i_1}
b_{jk\}}(i)=0, 
\]
then the solution develops logarithmic singularities at
$\tau=\pm 1$.
\end{itemize}
These qualitative properties of the solutions to the transport
equations are independent of whether one uses the $\kappa=\rho$ or the
$\kappa=\omega$ gauge. 
\end{proposition}

Detailed expressions for the independent components of the
spinor $\phi_{ABCD}$ up to order $p=2$ in the gauge $\kappa=\omega$ are
given in Appendix \ref{Appendix:Expansions}. We notice, however, that
\begin{subequations}
\begin{eqnarray}
&& \phi^{[2]}_0 = \mathfrak{b}_0 (1+\tau)^4 \bigg( \ln(1-\tau) -\ln(1+\tau)
   \bigg) + \breve{\phi}^{[2]}_0 , \label{Logarithms1}\\
&& \phi^{[2]}_1 = \mathfrak{b}_1 (1-\tau)(1+\tau)^3 \bigg(  \ln(1-\tau) -
   \ln(1+\tau) \bigg) + \breve{\phi}^{[2]}_1, \\
&& \phi^{[2]}_2= \mathfrak{b}_2 (1-\tau)^2(1+\tau)^2 \bigg( \ln(1-\tau) -
   \ln(1+\tau)\bigg) + \breve{\phi}^{[2]}_2,\\
&& \phi^{[2]}_3 = \mathfrak{b}_3 (1-\tau)^3(1+\tau) \bigg( \ln(1-\tau) -
   \ln(1+\tau) \bigg) + \breve{\phi}^{[2]}_3, \\
&& \phi^{[2]}_4 = \mathfrak{b}_4 (1-\tau)^4 \bigg(\ln(1-\tau) -
   \ln(1+\tau))\bigg) + \breve{\phi}^{[2]}_4, \label{Logarithms4}
\end{eqnarray}
\end{subequations}
where 
\[
\mathfrak{b}_k(i) \equiv b_{(\bmA\bmB\bmC\bmD)_k}(i)
\]
denotes the independent components of the Bach spinor (i.e. spinorial
counterpart of the 3-dimensional tensor $b_{ij}$) evaluated at the
point at infinity and $\breve{\phi^{[2]}}_0,\ldots,\;
\breve{\phi^{[2]}}_4$ denote expressions with polynomial dependence
(hence smooth) on $\tau$.  

\begin{remark}
{\em Observe that the coefficient $\phi_0^{[2]}$ is the most singular one
  at $\tau=1$  (i.e. $\mathscr{I}^+$), not being even continuous,
  while $\phi_4^{[2]}$ is the most regular having 4 derivatives at
  $\tau=1$. The role at $\tau=-1$ is reversed. }
\end{remark}

\begin{remark}
{\em That the structure of the logarithmic singularities in \eqref{Logarithms1}-\eqref{Logarithms4} suggest
  that, generically, the development of initial data satisfying
  Assumption \ref{Assumption:Initial data} does not admit a $C^2$
  conformal extension. This observation is likely to hold also for
  more general classes of initial data.}
\end{remark}

\medskip
In order to be able to relate the formal F-gauge expansions discussed
in the previous paragraphs we needs to make the following assumption:

\begin{assumption}
\label{Assumption:FormalExpansionsAreNotFormal}
The formal F-gauge expansions correspond to the leading orders of an
actual solution to the extended conformal Einstein field equations. In
particular, for the components of the rescaled Weyl tensor, one has that  
\[
\bmphi = \sum_{j=0}^2 \frac{1}{j!}\bmphi^{[j]} \rho^j + {\bm R}_3
\]
where the remainder satisfies ${\bm R}_3 \in C^\infty (\tilde{\mathcal{M}}_{a,\kappa})\cap
C^0(\mathcal{M}_{a,\kappa})$ and ${\bm R}_3=O(\rho^3)$. Similarly, for 
the components of the frame one has
\[
\bme = \sum_{j=0}^2 \frac{1}{j!}\bme^{[j]} \rho^j + {\bm R}_3.
\]
\end{assumption}
% \mnotex{JAVK (17.5.2017): one needs also assumptions on the expansions
% of the frame. Include these!}

\begin{remark}
{\em As already mentioned in the introduction, controlling the residue
of the F-expansions by means of estimates obtained from the conformal
evolution equations is the major outstanding issue in the so-called
\emph{problem of spatial infinity}. }
\end{remark}

\section{The NP gauge}
\label{Section:NPGauge}
While the F-gauge expressions discussed in the previous section
provide a great deal of information about the singular behaviour of
the Weyl tensor at the conformal boundary of the development of the class of initial data of
Assumption \ref{Assumption:Initial data}, they are given in a gauge
which does not directly lead to assertions about the peeling (or lack
thereof) behaviour of the spacetime. For this, one has to transform
into a gauge hinged on (future) null infinity ---the so-called
Newman-Penrose (NP) gauge. 

\medskip
In the following we use the notation $\simeq$ to denote equality at
$\mathscr{I}^+$. 

\subsection{Construction of the gauge}
In what follows consider a conformal extension $(\mathcal{M},\bmg',\Xi)$ of a \emph{suitably}
asymptotically spacetime satisfying the vacuum Einstein equations
containing at least a piece of future null infinity. A frame
$\{\bme'_{\bmA\bmA'}\}$ satisfying
$\bmg'(\bme'_{\bmA\bmA'},\bme'_{\bmB\bmB'})=\epsilon_{\bmA\bmB}\epsilon_{\bmA'\bmB'}$
defined in a neighbourhood $\mathcal{U}$ of $\mathscr{I}^+$ is said to
\emph{be adapted to} $\mathscr{I}^+$ if:
\begin{itemize}
\item[(i)] The vector $\bme'_{\bm1\bm1'}$ is tangent to
  $\mathscr{I}^+$ and is parallelly propagated along its generators.

\item[(ii)] On $\mathcal{U}$ there exists a function $u$ (a
  \emph{retarded time}) which is an affine parameter of the generators
  of $\mathscr{I}^+$ satisfying $\bme'_{\bm1\bm1'}(u)\simeq 1$. The
  function $u$ is propagated off the conformal boundary by requiring
  it to be constant on null hypersurfaces transverse to $\mathscr{I}^+$
  and satisfies $\bme'_{\bm0\bm0'}=\bmg^{\prime \sharp}(\mathbf{d}u,\cdot)$
  ---thus, $\bme'_{\bm0\bm0'}$ is tangent to the hypersurfaces
  $\mathscr{N}_{u_\bullet}$ defined by the condition $u=u_\bullet$
  where $u_\bullet$ is a constant.

\item[(iii)] The fields are tangent to the cuts
  $\mathcal{C}_{u_\bullet} \equiv \mathscr{N}_{u_\bullet} \cap
  \mathscr{I}^+$ and parallely propagated along the direction of $\bme'_{\bm0\bm0'}$.
\end{itemize}

The conditions (i)-(iii) imply restrictions on the form of several
components of the spin connection coefficients. Exploiting the
conformal freedom in the choice of the conformal factor $\Xi$, the
frame can be refined even further to fix the values of certain
components of the tracefree Ricci tensor of the metric $\bmg'$. More
precisely, one has the following:

\begin{proposition}
Given a suitably asymptotically simple spacetime, locally, it is always
possible to find a conformal extension $(\mathcal{M}',\bmg',\Xi)$ for
which
\[
R[\bmg']\simeq 0,
\]
and an adapted frame $\{ \bme'_{\bmA\bmA'}\}$ such that the associated
spin connection coefficients $\Gamma'_{\bmA\bmA'\bmB\bmC}$ satisfy 
\begin{eqnarray*}
&\Gamma'_{\bm0\bm0'\bmB\bmC}\simeq 0, \qquad
  \Gamma'_{\bm1\bm1'\bmB\bmC}\simeq 0,&\\
&\Gamma'_{\bm0\bm1'\bm1\bm1}\simeq 0, \qquad
  \Gamma'_{\bm1\bm0'\bm0\bm0}\simeq0, \qquad
  \Gamma'_{\bm1\bm0\bm1\bm1}\simeq 0,&\\
&\bar{\Gamma}'_{\bm1'\bm0\bm0'\bm1'}+ \Gamma_{\bm0\bm1'\bm0\bm1}\simeq
                                      0.&
\end{eqnarray*}
In addition, one has that
\[
\Phi'_{\bm12}\simeq 0, \qquad \Phi'_{22}\simeq 0,
\]
and $\bme'_{\bm0\bm0'}(\Xi)$ is constant on $\mathscr{I}^+$. 
\end{proposition}

The construction summarised in the previous proposition is
supplemented by adapted coordinates. On a fiduciary cut
$\mathscr{C}_\star\approx \mathbb{S}^2$ one chooses some coordinates
$\theta =(\theta^{\mathcal{A}})$, $\mathcal{A}=2,\,3$ and extends them
along $\mathscr{I}^+$ by requiring them to be constant along the
null generators. On the hypersurfaces $\mathscr{N}_u$ transverse to
$\mathscr{I}^+$ one identifies an affine parameter $r$ of the null
generators of these hypersurfaces in such a way that
$\bme'_{\bm0\bm0'}(r)=1$ and $r\simeq 0$. Finally, the coordinates
$\theta=(\theta^{\mathcal{A}})$ are propagated off $\mathscr{I}^+$ 
so that they are constant along the generators of
$\mathscr{N}_u$. From this construction one obtains \emph{unphysical
  Bondi coordinates} $x=(u,r,\theta^{\mathcal{A}})$ in a neighbourhood
of $\mathscr{I}^+$. \emph{Physical Bondi coordinates} are obtained
with the inversion $\tilde{r}=1/r$.

\subsection{Basic transformation formulae}
The relation between the F-gauge and the NP-gauge has been analysed in
\cite{FriKan00} ---the reader is referred to this article for full
details on this construction. The translation from the F-gauge to the
NP gauge involves:
\begin{itemize}
\item[(i)] a conformal rescaling of the metric,
\begin{equation}\label{FtoNPmetric}
 \bmg'=\varkappa^2\bmg
\end{equation}
where $\varkappa$ is a suitable conformal factor;

\item[(ii)]  a Lorentz transformation and rescaling of the frame $\{\bme_{\bmA\bmA'}\}$ of
  the form 
\begin{equation}
\label{FFrameToNPFrame}
\bme'_{\bmA\bmA'}=\varkappa^{-1}\Lambda^{\bmB}{}_{\bmA}
\bar{\Lambda}^{\bmB'}{}_{\bmA'}\bme_{\bmB \bmB'},
\end{equation}
where $(\Lambda^{\bmB}{}_{\bmA}) \in SL(2,\mathbb{C})$;

\item[(iii)] a change of coordinates
\[
u = u(\tau,\rho,s), \qquad r = r(\tau,\rho,s)
\]
where $u$ is a suitable retarded time on $\mathscr{I}^+$ and $r$ is an
affine parameter the generators of outgoing light cones such that $r|_{\mathscr{I}^+}=0$. 
\end{itemize}

As a consequence of relation \eqref{FFrameToNPFrame}, the spin dyads
$\{ \epsilon_{\bmA}{}^{A}\}$ and
$\{\epsilon^{\prime}{}_{\bmA}{}^{A}\}$, associated to the frames 
 $\{ \bme_{\bmA\bmA'}\}$ and $\{\bme'_{\bmA \bmA'}\}$, respectively, are related to
 each other via
\begin{equation}
\label{RelationSpinDyadsEqs}
\epsilon^{\prime}{}_{\bmA}{}^{A}=\varkappa^{-1/2}\Lambda^{\bmB}{}_{\bmA}\epsilon_{\bmB}{}^{A}.
\end{equation}

From the discussion in \cite{FriKan00} it readily follows that:

\begin{proposition}
\label{Proposition:LorentzExpanded}
For the development of initial data satisfying Assumption
\ref{Assumption:Initial data} and under Assumption
\ref{Assumption:FormalExpansionsAreNotFormal}  it follows that: 
 \begin{subequations}
\begin{eqnarray}
&& \Lambda^{\bm0}{}_{\bm1}=\rho^{1/2} + O\big(\rho^{7/2},(1-\tau)\big), \label{ExpansionLambdasAndConformalFactor1}
  \\ 
&& \Lambda^{\bm1}{}_{\bm1}= O\big(\rho^{5/2},
   (1-\tau)\big), \label{ExpansionLambdasAndConformalFactor2}\\ 
&&
  \Lambda^{\bm0}{}_{\bm0}=O\big(\rho^{3/2},(1-\tau)\big)
 \label{ExpansionLambdasAndConformalFactor3}\\ 
&&  \Lambda^{\bm1}{}_{\bm0}=-\rho^{-1/2}+ O\big(\rho^{3/2},(1-\tau)\big),
\end{eqnarray}
% \begin{eqnarray}
% && \Lambda^{\bm0}{}_{\bm1}=\rho^{1/2}\Big(1-\frac{3}{4}m\rho +
%   \bigg(\frac{15}{32}m^2-2W_{1}\bigg)\rho^2 +
%   O(\rho^3)\Big), \label{ExpansionLambdasAndConformalFactor1}
%   \\ && \Lambda^{\bm1}{}_{\bm1}=\rho^{5/2} \Big(\frac{1}{2}X_{+}W_{1}
%   + O(\rho)
%   \Big), \label{ExpansionLambdasAndConformalFactor2}\\ &&
%   \Lambda^{\bm0}{}_{\bm0}=\rho^{1/2}\Big(
%   -\frac{101}{10}X_{-}W_{1}\rho +
%   O(\rho^2)\Big), \label{ExpansionLambdasAndConformalFactor3}\\ &&
%   \Lambda^{\bm1}{}_{\bm0}=\rho^{-1/2}\Big(-1-\frac{3}{4}m\rho +O(\rho^2) \Big) \\ 
% \end{eqnarray}
\end{subequations}
and
\[
\varkappa= 1+ O\big(\rho,(1-\tau)\big).
\]
Moreover, one has that
\begin{eqnarray*}
&& u = \sqrt{2}\left( -\frac{1}{\rho} + 4m \ln \rho + u_\star\right) +
   O\big( \rho,(1-\tau) \big), \\
&& r = (1-\tau) + O\big(\rho,(1-\tau)^2\big).
\end{eqnarray*}
% \begin{eqnarray*}
% && u = \sqrt{2}\left( -\frac{1}{\rho} + 4m \ln \rho + u_\star +\bigg(
%    \frac{195}{28}m^2 + \frac{74}{5}W_1
%    \bigg)\rho + O(\rho^2)\right), \\
% && r = (1-\tau) + O\big(\rho,(1-\tau)^2\big).
% \end{eqnarray*}
\end{proposition}

\begin{remark}
{\em Observe that $\varkappa \neq 0$ on $\mathscr{I}^+$, thus, the NP
  gauge used in the previous discussion is expressed in the unphysical
spacetime.}
\end{remark}

\begin{figure}[t]
\centering
\includegraphics[width=0.9\textwidth]{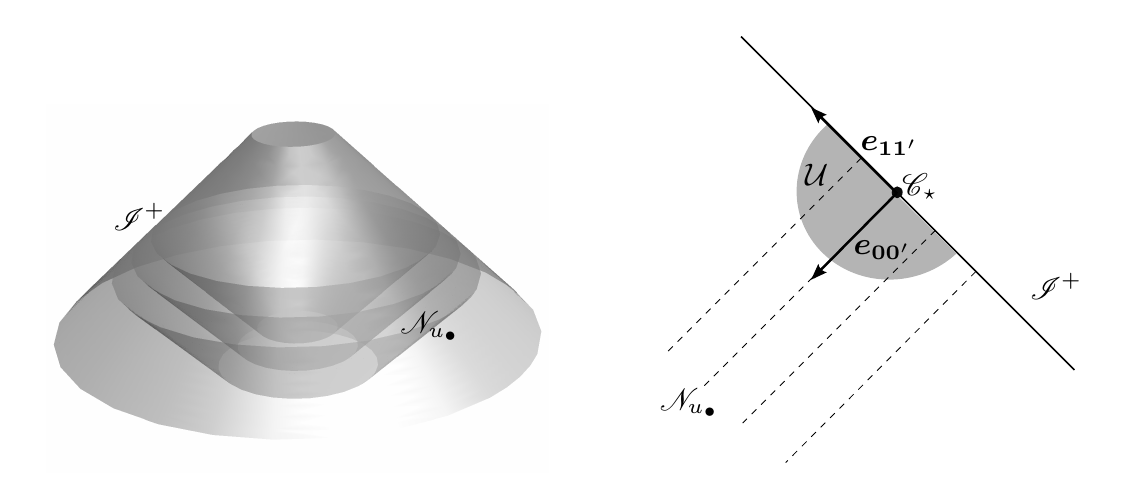}
\caption{Schematic depiction of the NP gauge. Left: a foliation of outgoing
  null hypersurfaces near future null infinity. Right: the vectors of
  the adapted frame.}
\label{Figure:NPGauge}
\end{figure}

\subsection{Transformation formulae for the rescaled Weyl spinor}
\label{Section:TransformationWeyl}

Consistent with equation \eqref{FtoNPmetric} one has that, 
the rescaled Weyl spinor associated with the $\bmg'$ and $\bmg$
representation denoted respectively 
 as $\phi_{ABCD}'$  and $\phi_{ABCD}$, are related to each
other via
\[
\phi'_{ABCD}=\varkappa^{-1} \phi_{ABCD}.
\]
Transvecting this last equation
with $\epsilon^{\prime}_{\bmA}{}^{A}$ and using equation
\eqref{RelationSpinDyadsEqs} one concludes that the relation between
the components of the rescaled Weyl spinors, respect to the
NP and F-frames, is given by
\begin{equation} 
\label{CompsNPWeylInTermsOfCompsFrescaledWeylGeneral}
\phi'_{\bmA \bmB \bmC \bmD} =\varkappa^{-3} \Lambda^{\bmF}{}_\bmA\Lambda^{\bmH}{}_\bmB\Lambda^{\bmP}{}_\bmC\Lambda^{\bmQ}{}_\bmD\phi_{\bmF \bmH \bmP \bmQ}.
\end{equation}
The components of the
rescaled Weyl spinor $\phi_{ABCD}$ in the F-gauge will be expressed
through contractions with the spin dyad $\epsilon_{\bmA}{}^{A}$. Following the standard conventions of \cite{PenRin86,Ste91} one defines:
\begin{eqnarray*}
&& \phi_{0} \equiv o^{A}o^{B}o^{C}o^{D}\phi_{ABCD},  \\
&& \phi_{1} \equiv   \iota^{A}o^{B}o^{C}o^{D}\phi_{ABCD}, \\ 
&& \phi_{2} \equiv \iota^{A}\iota^{B}o^{C}o^{D}\phi_{ABCD},  \\  
&& \phi_{3}\equiv \iota^{A}\iota^{B}\iota^{C}o^{D}\phi_{ABCD},\\
&& \phi_{4} \equiv  \iota^{A}\iota^{B}\iota^{C}\iota^{D}\phi_{ABCD}. 
\end{eqnarray*}
so that $\phi_{ABCD}$ can be expressed as
\[
\phi_{ABCD}=\phi_{0}\iota_{A}\iota_{B}\iota_{C}\iota_{D} -4 \phi_{1}\iota_{(A}\iota_{B}\iota_{C}o_{D)} + 6
 \phi_{2}\iota_{(A}\iota_{B}o_{C}o_{D)}-4  \phi_{3}\iota_{(A}o_{B}o_{C}o_{D)}+  \phi_{4}o_{A}o_{B}o_{C}o_{D}.
\]

% For the components of the Weyl spinor
% $\Psi_{ABCD}$ in the NP gauge, the conventions used in
% \cite{PenRin86,Ste91} are followed. Namely, one has that
% \begin{eqnarray*}
% &&\Psi_{0} \equiv o^{\prime A}o^{\prime B}o^{\prime C}o^{\prime D}\Psi_{ABCD},\\
% && \Psi_{1} \equiv
% \iota^{\prime A}o^{\prime B}o^{\prime C}o^{\prime D}\Psi_{ABCD}, \\
% && \Psi_{2}\equiv
% \iota^{\prime A}\iota^{\prime B}o^{\prime C}o^{\prime D}\Psi_{ABCD},  \\ 
% && \Psi_{3}\equiv
% \iota^{\prime A}\iota^{\prime B}\iota^{\prime C}o^{\prime
%    D}\Psi_{ABCD}, \\
% && \Psi_{4}\equiv
% \iota^{\prime A}\iota^{\prime B}\iota^{\prime C}\iota^{\prime D}\Psi_{ABCD}.
% \end{eqnarray*}
 
\noindent Expanding equation \eqref{CompsNPWeylInTermsOfCompsFrescaledWeylGeneral} and taking into the above introduced notation one obtains
 \begin{subequations}
\begin{eqnarray*}
 && \phi'_{0}= \varkappa^{-3} \Big((\Lambda^{\bm0}{}_{\bm0})^4
   \phi_{0} + 4 (\Lambda ^{\bm0}{}_{\bm0})^3(\Lambda ^{\bm1}{}_{\bm0})\phi_{1} +
   6(\Lambda^{\bm0}{}_{\bm0})^2(\Lambda^{\bm1}{}_{\bm0})^2\phi_{2} +
   4(\Lambda^{\bm0}{}_{\bm0})(\Lambda^{\bm1}{}_{\bm0})^3\phi_{3} +
   (\Lambda ^{\bm1}{}_{\bm0})^4\phi_{4}\Big), \\
&&  \phi'_{1} = \varkappa^{-3} \Big((\Lambda ^{\bm0}{}_{\bm0})^3(\Lambda^{\bm0}{}_{\bm1}) \phi_{0} +3(\Lambda^{\bm0}{}_{\bm0})^2(\Lambda ^{\bm0}{}_{\bm1})(\Lambda ^{\bm1}{}_{\bm0})  \phi_{1} +   (\Lambda ^{\bm0}{}_{\bm0})^3(\Lambda ^{\bm1}{}_{\bm1})\phi_{1}
    \\
&& \hspace{3cm}+ 3 (\Lambda ^{\bm0}{}_{\bm0})(\Lambda^{\bm0}{}_{\bm1})(\Lambda ^{\bm1}{}_{\bm0})^2   \phi_{2} + 3  (\Lambda^{\bm0}{}_{\bm0})^2(\Lambda^{\bm1}_{\bm0})(\Lambda^{\bm1}_{\bm1})\phi_{2}+   (\Lambda^{\bm0}{}_{\bm1})(\Lambda^{\bm1}_{\bm0})^3\phi_{3}
  \\ 
&&   \hspace{3cm}+3(\Lambda^{\bm0}{}_{\bm0})(\Lambda^{\bm1}{}_{\bm0})^2(\Lambda^{\bm1}{}_{\bm1})\phi_{3}
   + (\Lambda^{\bm1}{}_{\bm0})^3 (\Lambda^{\bm1}_{\bm1}) \phi_{4}
   \Big),\\
&& \phi'_{2}= \varkappa^{-3}\Big( (\Lambda^{\bm0}{}_{\bm0})^2 (\Lambda^{\bm0}{}_{\bm1})^2 \phi_{0} + 2
 (\Lambda^{\bm0}{}_{\bm0})(\Lambda^{\bm0}{}_{\bm1})^2(\Lambda^{\bm1}{}_{\bm0})\phi_{1}
   + 2(\Lambda^{\bm0}_{\bm0})^2(\Lambda^{\bm0}{}_{\bm1})
   (\Lambda^{\bm1}{}_{\bm1})\phi_{1} 
   \\ 
&& \hspace{3cm}+(\Lambda^{\bm0}{}_{\bm1})^2(\Lambda^{\bm1}{}_{\bm0})^2\phi_{2}
   +4(\Lambda^{\bm0}{}_{\bm0})(\Lambda^{\bm0}{}_{\bm1})
(\Lambda^{\bm1}{}_{\bm0})(\Lambda^{\bm1}{}_{\bm1})\phi_{2}
   + (\Lambda^{\bm0}{}_{\bm0})^2
   (\Lambda^{\bm1}{}_{\bm1})^2\phi_{2} \\ 
&& \hspace{3cm}+2(\Lambda^{\bm0}{}_{\bm1})(\Lambda^{\bm1}{}_{\bm0})^2(\Lambda^{\bm1}{}_{\bm1})\phi_{3}
  +   2(\Lambda^{\bm0}{}_{\bm0})(\Lambda^{\bm1}{}_{\bm0})(\Lambda^{\bm1}{}_{\bm1})^2
\phi_{3} + (\Lambda^{\bm1}{}_{\bm0})^2
  (\Lambda^{\bm1}{}_{\bm1})^2\phi_{4} \Big),\\
&& \phi'_{3}= \varkappa^{-3}\Big(
   (\Lambda^{\bm0}{}_{\bm0})(\Lambda^{\bm0}{}_{\bm1})^3\phi_{0} +
   (\Lambda^{\bm0}{}_{\bm1})^3(\Lambda^{\bm1}{}_{\bm0})\phi_{1}
   + 3(\Lambda^{\bm0}{}_{\bm0})(\Lambda^{\bm0}{}_{\bm1})^2(\Lambda^{\bm1}{}_{\bm1})\phi_{1}
  \\ 
&&
 \hspace{3cm}  +3(\Lambda^{\bm0}{}_{\bm1})^2(\Lambda^{\bm1}{}_{\bm0})(\Lambda^{\bm1}{}_{\bm1})\phi_{2} + 3(\Lambda^{\bm0}{}_{\bm0})(\Lambda^{\bm0}{}_{\bm1})(\Lambda^{\bm1}{}_{\bm1})^2\phi_{2} + 3(\Lambda^{\bm0}{}_{\bm1})(\Lambda^{\bm1}{}_{\bm0})(\Lambda^{\bm1}{}_{\bm1})^2\phi_{3} \\
&& \hspace{3cm}+
   (\Lambda^{\bm0}{}_{\bm0})(\Lambda^{\bm1}{}_{\bm1})^3\phi_{3} +
   (\Lambda^{\bm1}{}_{\bm0})(\Lambda^{\bm1}{}_{\bm1})^3\phi_{4}
   \Big),\\
&& \phi'_{4}= \varkappa^{-3} \Big((\Lambda^{\bm0}{}_{\bm1})^4\phi_{0} +
   4(\Lambda^{\bm0}{}_{\bm1})^3(\Lambda^{\bm1}{}_{\bm1})\phi_{1} +
   6(\Lambda^{\bm0}{}_{\bm1})^2(\Lambda^{\bm1}{}_{\bm1})^2\phi_{2} +
   4(\Lambda^{\bm0}{}_{\bm1})(\Lambda^{\bm1}{}_{\bm1})^3\phi_{3} +
   (\Lambda^{\bm1}{}_{\bm1})^4\phi_{4}\Big).
 \end{eqnarray*}
\end{subequations}

\noindent Finally, to obtain to compute the (physical) components of the Weyl
tensor we observe that
\begin{equation}
\bmg' = r^2 \tilde{\bmg}
\label{NPToPhysicalMetric}
\end{equation}
where $r$ is the (unphysical) Bondi radial coordinate. By definition
one has then that
\[
\phi'_{ABCD} = r \Psi_{ABCD}
\]
where $\Psi_{ABCD}$ is the (conformally invariant) Weyl
spinor. Consistent with the rescaling \eqref{NPToPhysicalMetric} one
sets
\begin{equation}
\label{NPToPhysicalDyad}
o^{\prime A} = r \tilde{o}^A, \qquad \iota^{\prime A} =\tilde{\iota}^A.
\end{equation}
The physical Bondi radial coordinate $\tilde{r}$ is obtained from $r$
via the inversion
\[
\tilde{r} = 1/r.
\] 

It follows from the above that the components $\tilde\psi_0,
\ldots,\tilde\psi_4$ of $\Psi_{ABCD}$ with respect to the spin dyad $\{
\tilde{o}^A,\tilde{\iota}^A\}$ are given by
\[
\tilde{\psi}_0 =  \frac{1}{\tilde{r}^3} \phi'_0, \qquad
\tilde{\psi}_1 = \frac{1}{\tilde{r}^2} \phi'_1, \qquad
\tilde{\psi}_2 =  \frac{1}{\tilde{r}} \phi'_2, \qquad
\tilde{\psi}_3 =\phi'_3, \qquad \tilde{\psi}_4 = \tilde{r} \phi'_4.
\]

\begin{remark}
{\em Observe that the rescaling of the spin dyad given in
  \eqref{NPToPhysicalDyad} is asymmetric in $o^A$ and $\iota^A$
  ---this choice is customary in the discussion of the peeling behaviour.}
\end{remark}

\section{Main results}
\label{Section:PolyhomogeneousExpansions}
In this section we provide the main results of our analysis:
expressions for the asymptotic decay of the components of the Weyl tensor dictated by
the F-expansions under the premise that these expansions are related
to actual solutions to the conformal Einstein field equations as
stated in Assumption
\ref{Assumption:FormalExpansionsAreNotFormal}. These results bring
together the discussion in Sections
\ref{Section:CylinderSpatialInfinity} and \ref{Section:NPGauge}. The statements in this section are obtained via a direct computation involving 
the transformation formulae of Section
\ref{Section:TransformationWeyl}, the explicit expansions of
Proposition \ref{Proposition:LorentzExpanded} and the solutions to the
transport equations as given by Proposition \ref{Proposition:ExpansionsCylinder}.

\subsection{Decay near null infinity}

%\subsection{General remarks}
% \mnotex{JAVK (29.5.2017): there is some technical stuff on how to
%   relate solutions in the bundle to solutions in spacetime which
%   should be mentioned here. Also generic statements about the
%   behaviour of the solutions at the conformal boundary ---in the
%   generic case the Weyl tensor is actually singular.}

% Before providing the main results of our analysis, we make some
% general observations which have been used in the calculations. 

% \medskip
% The transformation formulae between the coordinates $(\tau,rho,s)$ on
% $\tilde{\mathcal{M}})_{a,\kappa}$ and the unphysical Bondi coordinates
% $(u,r,\theta^{\mathcal{A}})$ directly allow to show that the
% components of the 

% \subsection{Main results}

The main result in this article is the following:

\begin{theorem}
\label{Theorem:Nonpeeling1}
Under Assumption \ref{Assumption:FormalExpansionsAreNotFormal}, given
time symmetric initial data satisfying Assumption
\ref{Assumption:Initial data}, the  F-expansions imply that
\begin{eqnarray*}
&& \tilde{\psi}_0 = O(\tilde{r}^{-3}\ln \tilde{r}), \\
&& \tilde{\psi}_1 = O(\tilde{r}^{-3}\ln \tilde{r}), \\
&& \tilde{\psi}_2 = O(\tilde{r}^{-3}\ln \tilde{r}), \\
&& \tilde{\psi}_3 = O(\tilde{r}^{-2}), \\
&& \tilde{\psi}_4 = O(\tilde{r}^{-1}).
\end{eqnarray*}
\end{theorem}

\begin{remark}
{\em A general framework for the discussion of polyhomogeneous
spacetimes with the above asymptotics is given in
\cite{ChrMacSin95}. Polyhomogeneous spacetimes satisfying the above
decay have been studied in \cite{Val99a} where a number of
their properties are discussed ---in particular the existence of
so-called \emph{logarithmic Newman-Penrose constants}.}
\end{remark}

\begin{remark}
{\em It is worth recalling that, for the purposes of the discussion of
  peeling properties, the key component is $\tilde\psi_0$ ---i.e. the
  most singular one. Making use of an asymptotic characteristic
  initial value problem it is possible to deduce the decay of the
  other components if that of $\tilde\psi_0$ is prescribed on a
  fiduciary outgoing light cone. Using these methods it is possible to
show that the Einstein field equations are, in fact, consistent with a
decay of the form $\tilde{\psi}_0 = O(\tilde{r}^{-3}\ln^N \tilde{r})$
with $N$ some positive integer ---see \cite{Val98}.}
\end{remark}

For more restricted classes of initial data one has the following:
\begin{theorem}
\label{Theorem:Nonpeeling2}
Under Assumption \ref{Assumption:FormalExpansionsAreNotFormal}, given
time symmetric initial data satisfying Assumption
\ref{Assumption:Initial data},  the F-expansions are such that:
\begin{itemize}
\item[(i)] If 
\[
b_{ij}(i)=0,
\]
  then
\begin{eqnarray*}
&& \tilde{\psi}_0 = O(\tilde{r}^{-4}\ln \tilde{r}), \\
&& \tilde{\psi}_1 = O(\tilde{r}^{-4}\ln \tilde{r}), \\
&& \tilde{\psi}_2 = O(\tilde{r}^{-3}), \\
&& \tilde{\psi}_3 = O(\tilde{r}^{-2}), \\
&& \tilde{\psi}_4 = O(\tilde{r}^{-1}).
\end{eqnarray*}

\item[(ii)] If
\[
b_{ij}(i)=0, \qquad D_{\{k} b_{ij\}}(i)=0, 
\]
then 
\begin{eqnarray*}
&& \tilde{\psi}_0 = O(\tilde{r}^{-5}\ln \tilde{r}), \\
&& \tilde{\psi}_1 = O(\tilde{r}^{-4}), \\
&& \tilde{\psi}_2 = O(\tilde{r}^{-3}), \\
&& \tilde{\psi}_3 = O(\tilde{r}^{-2}), \\
&& \tilde{\psi}_4 = O(\tilde{r}^{-1}).
\end{eqnarray*}

\item[(iii)] The classical peeling behaviour is obtained if
\[
b_{ij}(i)=0, \qquad D_{\{k} b_{ij\}}(i)=0, \qquad D_{\{k}D_{l}
b_{ij\}}(i)=0.
\]

\end{itemize}
\end{theorem}

\begin{remark}
{\em Polyhomogeneous spacetimes with asymptotics of the form (i) and
  (ii) in the previous result have been discussed in \cite{Val98}. The
relation to the \emph{outgoing radiation condition} of Bondi et
al. \cite{BonBurMet62,Sac62c} has been analysed in \cite{Val99b}. The
decay in (i) includes the \emph{minimal} polyhomogeneous spacetimes of
\cite{ChrMacSin95}. }
\end{remark}

\begin{remark}
{\em The peeling spacetimes of (iii) are not, generically, smooth at
  the conformal boundary. In fact, they will be exhibit the
  logarithmic singularities first observed in \cite{Val04a,Val04d}
  ---see also (iv) in Proposition \ref{Proposition:ExpansionsCylinder}. These
  will appear at order $O(\tilde{r}^{-6})$. The developments of
  explicit time symmetric initial data sets like those of
  Brill-Lindquist and Misner data should have this type of asymptotics
  ---i.e. peeling but with a conformal boundary of finite differentiability.}
\end{remark}

\subsection{Behaviour on null infinity}
The methods discussed in the previous sections also allow us to
analyse the behaviour of the rescaled Weyl tensor on $\mathscr{I}^+$
near spatial infinity. Observe that the restriction of
$\phi'_{\bmA\bmB\bmC\bmD}$ corresponds, essentially, to the leading
terms of $\tilde{\psi}_{\bmA\bmB\bmC\bmD}$. 

\begin{theorem}
Under Assumption \ref{Assumption:FormalExpansionsAreNotFormal}, given
time symmetric initial data satisfying Assumption
\ref{Assumption:Initial data},  the F-expansions imply that:
\begin{itemize}
\item[(i)] If $b_{ij}(i)\neq 0$ then
  $\phi'_{\bmA\bmB\bmC\bmD}|_{\mathscr{I}^+}$ are singular.

\item[(ii)] If $b_{ij}(i)= 0$ then one has that 
\[
\phi'_0 \simeq O(1), \qquad \phi'_1 \simeq O(1) \qquad \phi'_2 \simeq O(1)  \qquad \phi'_3 \simeq O(\rho^3)    \qquad \phi'_4\simeq O(\rho^3).
\]
\end{itemize}
\end{theorem}

\section{Construction of candidate spacetimes with non-peeling
  behaviour}
\label{Section:Candidates}
The purpose of this section is to construct an explicit family of
global spacetimes which are candidates for not satisfying the peeling
behaviour. In a first step we construct a family of asymptotically
Euclidean time symmetric
initial data sets which can be regarded as perturbations of data for
the Minkowski spacetime and such that $\mathfrak{b}_k\neq 0$. In a second step, we make use 
 of  Bieri's generalisation of Christodoulou \&
Klainerman's global existence and stability result \cite{Bie10}  to
guarantee the existence of a geodesically complete maximal hyperbolic
development of the members of the family of initial data sets. 

\subsection{Time symmetric initial data with prescribed asymptotics}

In this subsection we construct a family of time symmetric,
asymptotically Euclidean vacuum initial data sets such that
$\mathfrak{b}_k\neq 0$ ---hence the regularity condition \eqref{FriedrichRegularityCondition} is
violated at order zero. For simplicity, in what follows we consider time symmetric
initial data sets with one asymptotic end. 

\subsubsection{Construction of a suitable conformal metric}
We begin by recalling the connection between the Bach tensor and
normal expansions. As it is well known (see e.g. \cite{CFEBook}
Section 11.6.2), it is always possible to find, locally, 
a conformal gauge such that the conformal metric $\bmh$ admits in a
neighbourhood of $i$, in terms of normal coordinates $\underline{x}=(x^\alpha)$, the expansion
\begin{equation}
h_{\alpha\beta} = \delta_{\alpha\beta} -\frac{1}{6}  \partial_\eta
r_{\alpha\gamma\beta\delta}(i) x^\gamma x^\delta x^\eta + O( |x|^4).
\label{NormalExpansion}
\end{equation}
The relation between the value of the derivative of the Riemann tensor
of the metric $\bmh$ at $i$ and the value of the Bach tensor is best
discussed using spinors. Let
$r_{AGBHCDEF}$ denote the (space) spinor counterpart of the
3-dimensional Riemann tensor $r_{\alpha\gamma\beta\delta}$. As a
consequence of the symmetries of the Riemann tensor, its spinorial
counterpart admits the decomposition
\[
r_{AGBHCDEF} = -r_{ABCDEF} \epsilon_{GH} - r_{GHCDEF} \epsilon_{AB}
\]
where 
\[
r_{ABCDEF} = \bigg( \frac{1}{2} s_{ABCE} -\frac{1}{12}r h_{ABCE}
  \bigg)\epsilon_{DF} + \bigg( \frac{1}{2}s_{ABDF} - \frac{1}{12} h_{ABDF} \bigg)\epsilon_{CE},
\]
and $s_{ABCD}=s_{(ABCD)}$ denotes the spinorial counterpart of the
tracefree Ricci tensor and $r$ is the Ricci scalar. The spinor
$s_{ABCD}$ is related to the Bach spinor $b_{ABCD}$ via
\[
b_{ABCD} = D^Q{}_{(A} s_{BCD)Q}.
\]
Now, from the expansion \eqref{NormalExpansion} it follows that 
\[
r(i)=0, \qquad s_{ABCD}(i)=0.
\]
A computation then shows that 
\[
D_{AB} s_{CDEF} = -\frac{5}{6}\big( \epsilon_{A(C} b_{BDEF)} +
\epsilon_{B(C} b_{ADEF)} \big) \qquad \mbox{at} \quad i.
\]
Thus, if there is at least one non-vanishing essential component of
$b_{ABCD}(i)$ ---say $\mathfrak{b}_0$, then $(D_{AB}
s_{CDEF})(i)\neq0$ and accordingly, $(\partial_\eta
r_{\alpha\gamma\beta\delta})(i)\neq 0$. 

\medskip
Let now, $\mathfrak{z}_{\eta\alpha\gamma\beta\delta}$ denote real
constants, not all of them vanishing, satisfying the same symmetries
with respect to the indices $\alpha, \, \beta,\,\gamma,\, \delta, \,
\eta$ as the covariant derivative of the Riemann tensor. Moreover, let
$\bmhbar$ denote the standard metric of $\mathbb{S}^3$. The previous
discussion motivates the following:

\begin{lemma}
\label{Lemma:ConformalMetric}
Let $\mathcal{S}$ denote a compact 3-manifold with $\mathcal{S}\approx
\mathbb{S}^3$. Given constants
$\mathfrak{z}_{\eta\alpha\gamma\beta\delta}$ as above, there exists a
smooth Riemannian metric $\breve{\bmh}$ together with a smooth function $\omega$
over $\mathcal{S}$ and an $\varepsilon>0$ such that:
\begin{itemize}
\item[(i)] there exists a point $i\in \mathcal{S}$ such that in terms
  of normal coordinates $\underline{x}=(x^\alpha)$ centred at $i$ the metric $\breve{\bmh}$ takes the
  form
\begin{equation}
 \breve{h}_{\alpha\beta} = \delta_{\alpha\beta} +
 \mathfrak{z}_{\eta\alpha\gamma\beta\delta} x^\gamma x^\delta x^\eta
\label{MetricExample}
\end{equation}
in a $\breve{\bmh}$-open ball of radius $\varepsilon$,
$\mathcal{B}_\varepsilon(i)$, centred at $i$;

\item[(ii)] $\breve{\bmh}=\omega^2 \bmhbar$ on
  $\mathcal{S}\setminus\mathcal{B}_{2\varepsilon}(i)$;

\item[(iii)] the metric $\breve{\bmh}$ has positive Yamabe number
  $Y(\breve{\bmh})$;

\item[(iv)] one has that $\breve{\bmh}\rightarrow \omega^2 \bmhbar$ as
  $\varepsilon\rightarrow 0$ in the $C^2$ topology of smooth  metrics.  

\end{itemize}
\end{lemma}

\begin{proof}
The result is a direct application of Lemma  2.1 in \cite{Fri98a}. 
\end{proof}

\begin{remark} {\em From the expansion \eqref{NormalExpansion} one has
that for $\breve{\bmh}$ as given above,
\[
\partial_\eta\breve{r}_{\alpha\gamma\beta\delta}(i)=-6\mathfrak{z}_{\eta\alpha\gamma\beta\delta}.
\] Thus, as a consequence of \eqref{MetricExample} the metric
$\breve{\bmh}$ is not conformally flat in a small enough neighbourhood
of $i$. Moreover, it is clear that the constants
$\mathfrak{z}_{\eta\alpha\gamma\beta\delta}$ can be chosen so that
$\mathfrak{b}_0\neq 0$ to that the component of the Bach tensor
giving rise to the most singular behaviour of the rescaled Weyl tensor
at the critical set $\mathcal{I}^+$ is non-vanishing. }
\end{remark}

\subsubsection{Solving the time symmetric Hamiltonian constraint}
The family of metrics given by Lemma \ref{Lemma:ConformalMetric} will
now be used to construct time
symmetric initial data sets for the Einstein field equations via the
\emph{method of punctures} ---see \cite{Fri98a,BeiOMu94}. Accordingly, we
look for a conformal factor $\Omega$ such that $\Omega\in
C^2(\mathcal{S}) \cap C^\infty(\mathcal{S}\setminus\{ i\})$ satisfying
the properties
\begin{subequations}
\begin{eqnarray}
&&\Omega =0, \qquad \mathbf{d}\Omega =0, \qquad
  \mbox{\textbf{Hess}}\,\Omega = -2 \breve{\bmh} \qquad \mbox{at } \quad i, \label{BoundaryCondition1}\\
&& \Omega>0 \qquad \mbox{on} \quad \mathcal{S}\setminus \{  i\}, \label{BoundaryCondition2}
\end{eqnarray}
\end{subequations}
and
\begin{equation}
r[\Omega^{-2}\breve{\bmh}]=0.
\label{TimeSymmetricHamiltonianConstraint}
\end{equation}
Setting 
\[
\Omega =\vartheta^{-2},
\]
the time symmetric Hamiltonian constraint \eqref{TimeSymmetricHamiltonianConstraint} yields the Yamabe equation
\begin{equation}
\bigg( \Delta_{\breve{\bmh}} - \frac{1}{8}r[\breve{\bmh}]
  \bigg)\vartheta =0.
\label{YamabeEquation}
\end{equation}
Consistent with conditions
\eqref{BoundaryCondition1}-\eqref{BoundaryCondition2} we impose the
boundary condition
\begin{equation}
\vartheta |x|\rightarrow 1 \qquad \mbox{as} \quad x\rightarrow 0.
\label{BoundaryConditionALT}
\end{equation}
One then has the following:

\begin{proposition}
\label{Proposition:PunctureMethod}
Let $\breve{\bmh}$ and $\varepsilon>0$ be as given by Lemma
\ref{Lemma:ConformalMetric}. There exists a unique solution
$\vartheta$ to the Yamabe equation \eqref{YamabeEquation} with
boundary conditions \eqref{BoundaryConditionALT}. In
a suitably neighbourhood of $i$ the function $\vartheta$ can be
written as 
\[
\vartheta = \frac{U}{|x|} + W
\]
where $U$ and $W$ are analytic at $i$ and satisfy
\[
U= 1+ O(|x|^4), \qquad W(i)=\frac{m}{2}.
\]
The function $\vartheta$ depends smoothly on $\breve{\bmh}$ and, in
particular,
\[
\vartheta \rightarrow \frac{1}{|x|} \qquad \mbox{as} \quad
\varepsilon\rightarrow 0,
\]
so that
\[
\vartheta^4 \breve{\bmh} \rightarrow \bmdelta \qquad \mbox{as} \quad
\varepsilon\rightarrow 0,
\]
where $\bmdelta$ denotes the 3-dimensional Euclidean metric. 
\end{proposition}

\begin{proof}
The positivity of the Yamabe invariant of $\breve{\bmh}$ ensures the
solvability of equation \eqref{YamabeEquation} subject to the
condition \eqref{BoundaryConditionALT}. A detailed discussion of the
construction of the solution $\vartheta$ can be found in
\cite{Fri98a}, Section 2.2. Observe that in the limit
$\varepsilon\rightarrow 0$ one has that $m=0$ so from the rigidity
part of the mass positivity theorem \cite{SchYau79} the limiting initial data must be
data for the Minkowski spacetime. 
\end{proof}

\begin{remark}
{\em The limit $\varepsilon\rightarrow 0$ is independent of the choice
  of constants $\mathfrak{z}_{\eta\alpha\gamma\beta\delta}$. 
}
\end{remark}

\begin{remark}
{\em The decompactification obtained via the conformal factor
  $\vartheta$ gives rise to a metric $\tilde{\bmh}
  \equiv \vartheta^4 \breve{\bmh}$ defined over $\tilde{\mathcal{S}}\approx
    \mathbb{R}^3$ solving the time symmetric
  Hamiltonian constraint. By choosing $\varepsilon$ small enough one can make
    $\tilde{\bmh}$ suitably close to the 3-dimensional Euclidean metric. 
  }
\end{remark}

\begin{remark}
{\em The procedure describe in the previous paragraphs can be readily
  adapted to construct initial data sets for which $b_{ABCD}(i)=0$ but
  $D_{(EF} b_{ABCD)}(i)\neq 0$. 
}
\end{remark}

\begin{remark}
{\em A tedious calculation shows that if $b_{ABCD}(i)=0$ then the Bach
tensor of the metric $\tilde{\bmh}=\vartheta^4 \breve{\bmh}$,
expressed in terms of asymptotically Cartesian coordinates
$\underline{y}=(y^\alpha)$ in the asymptotic end of $\tilde{\mathcal{S}}$ satisfies
\begin{equation}
b_{\alpha\beta} =\frac{\mathfrak{b}_{\alpha\beta}}{|y|^4} +O(|y|^5),
\label{OurBachTensor}
\end{equation}
where $\mathfrak{b}_{\alpha\beta}$ are constants, not all of them zero.
}
\end{remark}

\subsection{Global existence of a candidate non-peeling spacetimes}
\label{Subsection:GlobalExistence}

In a next step we make use, for suitably small $\varepsilon>0$, of Bieri's
global existence and stability result of \cite{Bie10} to construct
the maximal globally hyperbolic development of the metric
$\tilde{\bmh}$ of Proposition \ref{Proposition:PunctureMethod}. The
resulting spacetime will be a candidate spacetime for exhibiting
non-peeling behaviour.  

\medskip
The analysis in \cite{Bie10} considers initial data sets
$(\tilde{\mathcal{S}},\tilde\bmh, \tilde\bmK)$,
$\tilde{\mathcal{S}}\approx \mathbb{R}^3$,  which outside a compact
set admit asymptotically Euclidean coordinates $\underline{x}=(x^\alpha)$ such
that:
\begin{subequations}
\begin{eqnarray}
&& \tilde{h}_{\alpha\beta} = \delta_{\alpha\beta} + O_{H^3} (|x|^{-1/2}), \label{CKAsymptotics1}\\
&& \tilde{K}_{\alpha\beta} = O_{H^2}(|x|^{-3/2}). \label{CKAsymptotics2}
\end{eqnarray}
\end{subequations}
% \begin{subequations}
% \begin{eqnarray}
% && \tilde{h}_{\alpha\beta} = \left( 1+ \frac{2m}{|x|}
%    \right)\delta_{\alpha\beta} + O_{H^4} (|x|^{-3/2}), \label{CKAsymptotics1}\\
% && \tilde{K}_{\alpha\beta} = O_{H^3}(|x|^{-5/2}). \label{CKAsymptotics2}
% \end{eqnarray}
% \end{subequations}
\footnote{One says that $u=O_{H^k}(r^{\alpha})$ if the weighted
  Sobolev norm $|| u ||_{H^{\alpha}_k}$ is finite ---see Appendix A in
  \cite{BieChr16} for detailed definition. In particular, if $k>l +
  3/2$ then $|\partial^l u| = o(r^\alpha)$. }
The metric $\tilde\bmh$ obtained in Proposition \ref{Proposition:PunctureMethod} can
be readily seen to satisfy the above asymptotic conditions ---in
particular $\tilde{K}_{\alpha\beta}=0$. 

In what follows $x_\odot\in \tilde{\mathcal{S}}$ is an arbitrary
origin in the (physical) hypersurface $\tilde{\mathcal{S}}$. Let
$d_\odot$ denote the distance function from the origin $x_\odot$. The
global smallness assumption in Bieri's global existence result makes
use of the quantity
\begin{eqnarray*}
&& Q( x_\odot,a) \equiv a^{-1} \int_{\mathbb{R}^3} \Big( |K|^2 + (a^2
   +d_\odot^2)|D\bmK|^2 + (a^2
   +d_\odot^2)^2 |D^2\bmK|^2  \\
&& \hspace{4cm} +(a^2 + d_\odot^2)  |\mbox{Ric}[\tilde\bmh]|^2 + (a^2 +
   d_\odot^2)^2 |D\mbox{Ric}[\tilde\bmh]|^2\Big)\mbox{d}\mu.
\end{eqnarray*}
with $a\in \mathbb{R}^+$ and where $\mbox{Ric}[\tilde\bmh]$ denotes
the Ricci tensor of the metric $\tilde{\bmh}$. The global existence and stability result in
\cite{Bie10} can be phrased as ---see \cite{BieChr16}:

% % % The global smallness assumption in Christodoulou \&
% % % Klainerman's global existence result make use of the quantity
% % \begin{eqnarray*}
% % && Q( x_\odot,b) \equiv \sup_{\mathbb{R}^3} \bigg( b^{-2}(d_\odot^2 + b^2)|
% %   \mbox{Ric}[\tilde\bmh]|^2\bigg) \\
% % && \hspace{3cm} + b^{-3} \int_{\mathbb{R}^3} \left(
% %    \sum_{l=0}^3(d_\odot^2+b^2)|D^l \bmK|^2 + \sum_{l=0}^1 (d_\odot^2+b^2)^{l+3}
% %    |D^l \mbox{Cotton}[\tilde\bmh]|^2 \right),
% % \end{eqnarray*}
% for $b\in \mathbb{R}^+$ and where $\mbox{Ric}[\tilde\bmh]$ and
% $\mbox{Cotton}[\tilde\bmh]$ denote, respectively, the Ricci and Cotton
% tensors of the metric $\tilde{\bmh}$. The global existence and stability result in
% \cite{ChrKla93} can be phrased as ---see \cite{BieChr16}:
\begin{theorem}
\label{Theorem:CK}
There exists $\epsilon>0$ such that for all smooth vacuum initial
data sets with \eqref{CKAsymptotics1}-\eqref{CKAsymptotics2}, $K_\alpha{}^\alpha=0$ and
satisfying 
\[
\inf_{x_\odot\in\mathbb{R}^3, \; a>0} Q(x_\odot,a) <\epsilon,
\]
the associated maximal globally hyperbolic vacuum development is
geodesically complete, with the metric asymptotically approaching the
Minkowski metric in all directions.  
\end{theorem}

Now, as $\tilde{\bmh}\rightarrow \bmdelta$ for $\varepsilon\rightarrow
0$,  and $\inf_{x_\odot\in\mathbb{R}^3, \; a>0} Q(x_\odot,a)=0$
for time symmetric Minkowski data, one can choose the $\varepsilon>0$
in Proposition \ref{Proposition:PunctureMethod} so
that $\inf_{x_\odot\in\mathbb{R}^3, \; a>0} Q(x_\odot,a)$ is small
enough for Theorem \ref{Theorem:CK} to apply. Summarising, one has the
following:

\begin{proposition}
\label{Proposition:NonpeelingExample}
For suitably small $\varepsilon>0$, the maximal globally hyperbolic
development, $(\tilde{\mathcal{M}},\tilde\bmg)$, of the time symmetric initial data set
$(\tilde{\mathcal{S}},\tilde\bmh)$ is geodesically complete with
$\tilde\bmg$ asymptotically approaching the Minkowski metric in all directions. 
\end{proposition}

Based on the results of Section
\ref{Section:PolyhomogeneousExpansions} we have the following:

\begin{conjecture}
\label{Conjecture:Polyhomogeneous}
The spacetime $(\tilde{\mathcal{M}},\tilde\bmg)$ obtained in
Proposition \ref{Proposition:NonpeelingExample} has polyhomogeneous
(non-peeling) asymptotics as in Theorem \ref{Theorem:Nonpeeling1}. 
\end{conjecture}

\begin{remark}
{\em The construction in the previous paragraphs can be readily adapted to
  obtain global spacetimes with polyhomogeneous (non-peeling) asymptotics as in
  Theorem \ref{Theorem:Nonpeeling2} part (i) or the peeling
  asymptotics of part (ii).} 
\end{remark}

\begin{remark}
\label{Remark:CannotUseCK}
{\em It should be pointed out that Christodoulou \& Klainerman's
global existence result in \cite{ChrKla93} cannot be applied to the
time symmetric initial data $(\tilde{\mathcal{S}},\tilde{\bmh})$ as
the quantity measuring the global smallness in this case requires $
d_\odot^6|\mbox{Bach}[\tilde\bmh]|^2$ to be integrable over
$\mathbb{R}^3$ ---this is not satisfied by a Bach tensor of the form
given by \eqref{OurBachTensor}.}
\end{remark}

\section{Concluding remarks}
\label{Section:Conclusions}
The main analysis of this article concerning the potential non-peeling
behaviour of the development of time symmetric initial data sets can,
in principle, be extended to more general classes of data for which
the extrinsic curvature is non-vanishing. Preliminary computations
and the general structure of the conformal Einstein field equations
suggest that the main picture will remain the same
---in particular, the generic non-peeling behaviour seems to be that
of Theorem \ref{Theorem:Nonpeeling1}. Its of interest to remark
that an analysis of the peeling properties of the development of
(conformally flat) Bowen-York initial data sets has been given in
\cite{Val07a} ---notice however, that because of the conformal flatness of
this family of initial data, the regularity condition
\ref{FriedrichRegularityCondition} is automatically satisfied and the
spacetimes have much better peeling properties. 

%%%%%%%%%%%%%%%%%%% Revised up to here! (2.6.2017) %%%%%%%%%%%%%%%%%%%%%%%%%%%

\medskip
As already mentioned in the introduction, the key open problem in
the analysis of the \emph{problem of spatial infinity} through the
cylinder at spatial infinity is to establish the relation between the
F-expansions obtained by solving the transport equations along
$\mathcal{I}$ and actual solutions to the conformal Einstein field
equations. More precisely, to assert that the F-expansions correspond
to a Taylor-like polynomial of a solution one needs to be able to
control the remainder of the expansion. A type of estimates that, in
principle, should provide the desired control have been constructed in
\cite{Fri03b} for the (linear) spin-2 massless field on the Minkowski
spacetime. Adapting this construction to the case of the full
conformal Einstein field equations is a challenging endeavour which
requires a much bigger insight into the structure of the equations
than the one it is currently available ---however, some preliminary
investigations have been carried out in \cite{Val09a}. Further
investigations along these lines will be pursued elsewhere. Needless
to say that the development of techniques to transform, say,
Conjecture \ref{Conjecture:Polyhomogeneous} into rigorous statement
would constitute a major milestone in mathematical General Relativity
and the conclusion of a research programme started in the early
1960's.

\section*{Acknowledgements}
We thank Piotr Chrusciel for encouraging us to
revisit the problem discussed in this article and for suggesting an
outline of the arguments in Section \ref{Subsection:GlobalExistence}.
We thank the Gravitational Physics Group of the University of Vienna
for its hospitality. The computation of the transformation formulae
for the Weyl tensor have been obtained using the suite {\tt xAct} for
tensorial and spinorial computations in {\tt Mathematica} ---see
\cite{xAct}. 

\appendix
\section{Expressions in the gauge $\kappa=\omega$}
\label{Appendix:Expansions}

In this appendix we present the explicit solutions to the transport
equations at $\mathcal{I}$ for the orders $p=0, \, 1, \, 2$. In view of
the applications in this article we restrict the attention to the
coefficients of the frame and the components of the rescaled Weyl spinor.

\subsection{Expansions for the frame coefficients}

In what follows let 
\[
(c^\mu_{\bmA\bmB})^{[p]} \equiv (\partial^p_\rho c^\mu_{\bmA\bmB})|_{\mathcal{I}},
\]
and 
\[
x_{\bmA\bmB} \equiv \sqrt{2} \delta_{(\bmA}{}^\bmzero
\delta_{\bmB)}{}^\bmone, \qquad y_{\bmA\bmB} \equiv
-\frac{1}{\sqrt{2}} \delta_\bmA{}^\bmone \delta_\bmB{}^\bmone,
\qquad z_{\bmA\bmB} \equiv \frac{1}{\sqrt{2}}\delta_\bmA{}^\bmzero \delta_\bmB{}^\bmone.
\]

\subsubsection{Order $p=0$}
\begin{eqnarray*}
&& (c^0_{\bmA\bmB})^0 = -\tau x_{\bmA\bmB}, \\
&& (c^1_{\bmA\bmB})^0 = 0, \\
&& (c^+_{\bmA\bmB})^0 = z_{\bmA\bmB}, \\
&& (c^-_{\bmA\bmB})^0 = y_{\bmA\bmB}.
\end{eqnarray*}

\subsubsection{Order $p=1$}
\begin{eqnarray*}
&& (c^0_{\bmA\bmB})^1 =m (-\tau+\tfrac{4}{3}\tau^3 -\tfrac{1}{3}\tau^5)x_{\bmA\bmB}, \\
&& (c^1_{\bmA\bmB})^1 = x_{\bmA\bmB}, \\
&& (c^+_{\bmA\bmB})^1 = m(\tfrac{1}{2}+\tau^2 -\tfrac{1}{6}\tau^4)z_{\bmA\bmB}, \\
&& (c^-_{\bmA\bmB})^1 = m(\tfrac{1}{2}+\tau^2 -\tfrac{1}{6}\tau^4)y_{\bmA\bmB}.
\end{eqnarray*}

\subsubsection{Order $p=2$}
\begin{eqnarray*}
&& (c^0_{\bmA\bmB})^2 = \big( m^2( 4
   \tau^3-5\tau^5+\tfrac{8}{7}\tau^7-\tfrac{1}{7}\tau^9) + W_1
   (-12\tau+16\tau^3-\tfrac{26}{5}\tau^5+\tfrac{6}{5}\tau^7) \big)
   x_{\bmA\bmB}\\
&& \hspace{3cm}  +
   (4\tau-8\tau^3+\tfrac{7}{5}\tau^5+\tfrac{3}{5}\tau^7)(X_-W_1 y_{\bmA\bmB} + X_+W_1 z_{\bmA\bmB}), \\
&& (c^1_{\bmA\bmB})^2 = m( 1-4\tau^2+\tfrac{2}{3}\tau^4) x_{\bmA\bmB}, \\
&& (c^+_{\bmA\bmB})^2 = \big( m^2(
   \tau^2+\tfrac{13}{6}\tau^4-\tfrac{8}{9}\tau^6+\tfrac{1}{14}\tau^8)
   +W_1( 4+12\tau^2-3\tau^4+\tfrac{3}{5}\tau^6)  \big) z_{\bmA\bmB} \\
&& \hspace{3cm}- (6\tau^2+\tfrac{1}{2}\tau^4-\tfrac{3}{10}\tau^6)
   X_-W_1 x_{\bmA\bmB}, \\
&& (c^-_{\bmA\bmB})^2 =\big( m^2(
   \tau^2+\tfrac{13}{6}\tau^4-\tfrac{8}{9}\tau^6+\tfrac{1}{14}\tau^8)
   +W_1(4+12\tau^2-3\tau^4+\tfrac{3}{5}\tau^6)  \big) z_{\bmA\bmB} \\
&& \hspace{3cm} -(6\tau^2+\tfrac{1}{2}\tau^4-\tfrac{3}{10}\tau^6)
   X_+W_1 x_{\bmA\bmB}.
\end{eqnarray*}

\subsection{Expansions of the rescaled Weyl tensor}
% \mnotex{JAVK $\mapsto$EGG (2.6.2017): harmonise the conventions with
%   the main text.\\
% EGG: taking into account these  observations  we have
% $\iota_{A}\epsilon^{A}{}_{\bmA}=-\epsilon_{\bmA}{}^{\bm0}$ $o_{A}\epsilon^{A}{}_{\bmA}=\epsilon_{\bmA}{}^{\bm1}$
%  then using the expression for $\phi_{ABCD}$ in abstract index notation given in the main text we get the expression for $\phi_{\bmA\bmB\bmC\bmD}$.}
Consistent with the notation introduced above let 
\[
(\phi_{\bmA\bmB\bmC\bmD})^{[p]} \equiv (\partial^p_\rho \phi_{\bmA\bmB\bmC\bmD})|_{\mathcal{I}}.
\]
Taking into account that 
\begin{equation*}
\epsilon^{\bm0}{}_{A}=-\iota_{A}, \qquad \epsilon^{\bm1}{}_{A}=o_{A},
\end{equation*}
and 
\begin{equation*}
\epsilon_{\bmA}{}^{A}\epsilon^{\bmB}{}_{A}=\delta_{\bmA}{}^{\bmB}.
\end{equation*}
the components $\phi_{\bmA\bmB\bmC\bmD}$ can be expressed as
\begin{eqnarray*}
&& \phi_{\bmA\bmB\bmC\bmD}=\phi_{0}\delta_{\bmA}{}^{\bm0}\delta_{\bmB}{}^{\bm0}\delta_{\bmC}{}^{\bm0}\delta_{\bmD}{}^{\bm0} + 4 \phi_{1}\delta_{(\bmA}{}^{\bm0}\delta_{\bmB}{}^{\bm0}\delta_{\bmC}{}^{\bm0}\delta_{\bmD)}{}^{\bm1} 
\\ && \hspace{3cm}+ 6
 \phi_{2}\delta_{(\bmA}{}^{\bm0}\delta_{\bmB}{}^{\bm0}\delta_{\bmC}{}^{\bm1}\delta_{\bmD)}{}^{\bm1}  + 4  \phi_{3}\delta_{(\bmA}{}^{\bm0}\delta_{\bmB}{}^{\bm1}\delta_{\bmC}{}^{\bm1}\delta_{\bmD)}{}^{\bm1} + \phi_{4}\delta_{\bmA}{}^{\bm1}\delta_{\bmB}{}^{\bm1}\delta_{\bmC}{}^{\bm1}\delta_{\bmD}{}^{\bm1}.
\end{eqnarray*}
where
\begin{eqnarray*}
&& \phi_{0} \equiv o^{A}o^{B}o^{C}o^{D}\phi_{ABCD},  \\
&& \phi_{1} \equiv   \iota^{A}o^{B}o^{C}o^{D}\phi_{ABCD}, \\ 
&& \phi_{2} \equiv \iota^{A}\iota^{B}o^{C}o^{D}\phi_{ABCD},  \\  
&& \phi_{3}\equiv \iota^{A}\iota^{B}\iota^{C}o^{D}\phi_{ABCD},\\
&& \phi_{4} \equiv  \iota^{A}\iota^{B}\iota^{C}\iota^{D}\phi_{ABCD}. 
\end{eqnarray*}

\subsubsection{Order $p=0$}
% \mnotex{EGG: changed using the relation between old and new notation for components $\phi$. Check old notes to see relation. We only have to multiply by $1/6$ whenever we see $\phi_{2}$, by $1/4$ whenever we see $\phi_{1}$ or $\phi_{3}$ in the old version. }
\[
\phi^{(0)}_{\bmA\bmB\bmC\bmD} = -m \delta_{(\bmA}{}^{\bm0}\delta_{\bmB}{}^{\bm0}\delta_{\bmC}{}^{\bm1}\delta_{\bmD)}{}^{\bm1}.
\]

\subsubsection{Order $p=1$}
\begin{multline*}
 \phi^{(1)}_{\bmA\bmB\bmC\bmD} = -\frac{1}{6}\big(
   36W_1(1-\tau^2)+m^2(9+18\tau^2-3\tau^4) \big)
   \delta_{(\bmA}{}^{\bm0}\delta_{\bmB}{}^{\bm0}\delta_{\bmC}{}^{\bm1}\delta_{\bmD)}{}^{\bm1}
 \\ -3 (1-\tau)^2X_+W_{1}
   \delta_{(\bmA}{}^{\bm0}\delta_{\bmB}{}^{\bm0}\delta_{\bmC}{}^{\bm0}\delta_{\bmD)}{}^{\bm1}
 +3(1+\tau)^2 X_-W_1
   \delta_{(\bmA}{}^{\bm0}\delta_{\bmB}{}^{\bm1}\delta_{\bmC}{}^{\bm1}\delta_{\bmD)}{}^{\bm1}. 
\end{multline*}

\subsubsection{Order $p=2$}

\begin{eqnarray*}
&&\phi^{(2)}_{\bmA\bmB\bmC\bmD}=\phi_{0}^{(2)}\delta_{\bmA}{}^{\bm0}\delta_{\bmB}{}^{\bm0}\delta_{\bmC}{}^{\bm0}\delta_{\bmD}{}^{\bm0} + 4 \phi_{1}^{(2)}\delta_{(\bmA}{}^{\bm0}\delta_{\bmB}{}^{\bm0}\delta_{\bmC}{}^{\bm0}\delta_{\bmD)}{}^{\bm1} 
\\ 
&&\hspace{3cm}+ 6
 \phi_{2}^{(2)}\delta_{(\bmA}{}^{\bm0}\delta_{\bmB}{}^{\bm0}\delta_{\bmC}{}^{\bm1}\delta_{\bmD)}{}^{\bm1} + 4  \phi_{3}^{(2)}\delta_{(\bmA}{}^{\bm0}\delta_{\bmB}{}^{\bm1}\delta_{\bmC}{}^{\bm1}\delta_{\bmD)}{}^{\bm1} + \phi_{4}^{(2)}\delta_{\bmA}{}^{\bm1}\delta_{\bmB}{}^{\bm1}\delta_{\bmC}{}^{\bm1}\delta_{\bmD}{}^{\bm1}
\end{eqnarray*}
where
\begin{eqnarray*}
&& \phi_0^{(2)} =  -2X_+^2 W_2 (1-\tau)^4 +
   \frac{\sqrt{2}}{3}b_0 f_0(\tau), \\
&& \phi_1^{(2)} = \frac{1}{4}a_1(\tau) m X_+W_1 -2 X_+W_2 (1+\tau)(1-\tau)^3+
   \frac{\sqrt{2}}{6}b_1 f_1(\tau),\\
&& \phi_2^{(2)} = \frac{1}{6}c_2(\tau) m^3+ \frac{1}{6}a_2(\tau) m W_1 -4W_2
   (1+\tau)^2(1-\tau)^2+ \frac{\sqrt{3}}{9} b_2 f_2(\tau),\\
&& \phi_3^{(2)} = \frac{1}{4}a_3(\tau) m X_-W_1  +2 X_-W_2
   (1+\tau)^3(1-\tau)+\frac{\sqrt{2}}{6}b_3 f_3(\tau),\\
&& \phi_4^{(2)} = -2X_-^2 W_2 (1+\tau)^4 +\frac{\sqrt{2}}{3}b_4 f_4(\tau),\\
\end{eqnarray*}
where
\begin{eqnarray*}
&& c_2(\tau) = -9-54\tau^2-69\tau^4+28\tau^6 -\frac{16}{7}\tau^8,\\
&& a_1(\tau) =-36+96\tau -132\tau^2 +184\tau^3 -82\tau^4 -\frac{164}{5}\tau^5+\frac{74}{5}\tau^6,\\
&& a_2(\tau) = -180-252\tau^2+372\tau^4-\frac{276}{5}\tau^6,\\
&& a_3(\tau) = -36-96\tau -132\tau^2 -184\tau^3 -82\tau^4 +\frac{164}{5}\tau^5+\frac{74}{5}\tau^6,\\
&& f_0(\tau) = 2(1-\tau)^4 K(-\tau), \\
&& f_1(\tau) = 4(1-\tau)^3(1+\tau) K(-\tau) - \frac{3}{1-\tau},\\
&& f_2(\tau) = \sqrt{6}\left(
   \frac{2-\tau}{(1+\tau)^2}-2(1-\tau)^2(1+\tau)^2K(\tau) \right),\\
&& f_3(\tau) = -4(1+\tau)(1-\tau)^3 K(\tau) + \frac{3}{1+\tau}, \\
&& f_4(\tau) = -2(1+\tau)^4K(\tau),
\end{eqnarray*}
and
\begin{eqnarray*}
&& K(\tau) = 1- 3\int_0^\tau \frac{\mbox{d}s}{(1-s)(1+s)^5}\\
&& \phantom{K(\tau)} = \frac{1}{32}\left( 3\ln(1-\tau) -3\ln(1+\tau)
   + \frac{32+38\tau+24\tau^2+6\tau^3}{(1+\tau)^4} \right)
\end{eqnarray*}
while
\[
b_k \equiv b_{(\bmA\bmB\bmC\bmD)_k}, \qquad k=0,\,1,\,\ldots 4,
\]
are the independent components of the Bach spinor evaluated at $i$. 

% QM 
%\bibliographystyle{reporthack}
% Ludovica
%\bibliographystyle{/Users/Juan/Documents/tex/reporthack}

% Path in QM 
%\bibliography{ThesisGRbib}
% Path in Ludovica
%\bibliography{/Users/Juan/Documents/tex/Newgrbib}

\end{document}